\documentclass[11pt, reqno]{amsart}
\usepackage{amsfonts,latexsym,enumerate}
\usepackage{amsmath}
\usepackage{amscd}
\usepackage{float,amsmath,amssymb,mathrsfs,bm,multirow,graphics}
\usepackage[dvips]{graphicx}
\usepackage[percent]{overpic}
\usepackage[pdftex]{color}
\usepackage{amsaddr}
\usepackage[numbers,sort&compress]{natbib}
\usepackage{enumitem}
\usepackage{todonotes}
\usepackage{mathtools}

\allowdisplaybreaks

\def\Xint#1{\mathchoice
{\XXint\displaystyle\textstyle{#1}}%
{\XXint\textstyle\scriptstyle{#1}}%
{\XXint\scriptstyle\scriptscriptstyle{#1}}%
{\XXint\scriptscriptstyle\scriptscriptstyle{#1}}%
\!\int}
\def\XXint#1#2#3{{\setbox0=\hbox{$#1{#2#3}{\int}$}
\vcenter{\hbox{$#2#3$}}\kern-.5\wd0}}

\def\pvint{\,\,\Xint-}

\newcommand{\ii}{{\rm i}}
\newcommand{\dd}{{\rm d}}

\newcommand{\im}{\mathrm{Im\,}}
\newcommand{\re}{\mathrm{Re\,}}

\newcommand{\Res}{\operatorname{Res}}
\newcommand{\res}{\operatorname{Res}}

\newcommand{\R}{{\mathbb R}}
\newcommand{\C}{{\mathbb C}}
\newcommand{\Z}{{\mathbb Z}}

\renewcommand{\Im}{\mathrm{Im}\hspace{0.09em}}

\newcommand{\csch}{\mathrm{csch}}

\newtheorem{theorem}{Theorem}
\newtheorem{proposition}{Proposition}[section]
\newtheorem{corollary}[proposition]{Corollary}
\newtheorem{lemma}[proposition]{Lemma}

\newtheorem{remark}[proposition]{Remark}

\numberwithin{equation}{section}

\newenvironment{roster}
 {\begin{enumerate}[font=\upshape,label=\Alph*.]}
 {\end{enumerate}}

\begin{document}

\title{On the non-chiral intermediate long wave equation}

\author{Bjorn K. Berntson$^1$, Edwin Langmann$^2$, and Jonatan Lenells$^1$}
\address{$^1$Department of Mathematics, KTH Royal Institute of Technology, SE-100 44 Stockholm, Sweden \\
$^2$Department of Physics, KTH Royal Institute of Technology, SE-106 91 Stockholm, Sweden}

\begin{abstract}
We study integrability properties of the non-chiral intermediate long wave equation recently introduced by the authors as a parity-invariant variant of the intermediate long wave equation. For this new equation we: (a) derive a Lax pair, (b) derive a Hirota bilinear form, (c) derive a B\"{a}cklund transformation, (d) use, separately, the B\"{a}cklund transformation and the Lax representation to obtain an infinite number of conservation laws. 
\end{abstract}

\maketitle



\tableofcontents

%
%
%
%
%

\section{Introduction}
The intermediate long wave (ILW) equation is a model for long gravity waves in a stratified, finite-depth fluid \cite{joseph1977}. The model is integrable: it has a Lax pair \cite{kodama1981,kodama1982}, a Hirota bilinear form and $N$-soliton solutions \cite{matsuno1979}, a B\"{a}cklund transformation \cite{satsuma1979}, and an infinite number of conservation laws \cite{satsuma1979}.

In this paper we consider the following ILW-type equation which was recently introduced by us in \cite{berntson2020a}:
\begin{equation} 
\label{2ilw} 
\begin{split} 
&u_t + 2 u u_x + Tu_{xx}+\tilde{T}v_{xx}=0,\\
&v_t - 2 v v_x - Tv_{xx}-\tilde{T}u_{xx}=0, 
\end{split} 
\end{equation} 
where $u(x,t)$ and $v(x,t)$ are  real- or complex-valued functions of $x,t \in \R$ and
\begin{equation} 
\label{TT} 
\begin{split} 
& (Tf)(x) \coloneqq \frac{1}{2\delta}\pvint_{\R} \coth\left(\frac{\pi}{2\delta}(x'-x)\right)f(x')\,\dd{x}',\\
& (\tilde{T}f)(x) \coloneqq \frac{1}{2\delta}\int_{\R} \tanh\left(\frac{\pi}{2\delta}(x'-x)\right)f(x')\,\dd{x}', 
\end{split} 
\end{equation} 
with $\delta>0$ an arbitrary parameter. 
We refer to \eqref{2ilw} as the non-chiral ILW equation because it is invariant under the parity transformation $[u(x,t),v(x,t)]\to[v(-x,t),u(-x,t)]$ (the standard ILW equation and its degenerations are not parity-invariant). 
A quantum version of the non-chiral ILW equation arises from a second-quantization of the elliptic Calogero-Sutherland model \cite{berntson2020a}, but here we consider only mathematical aspects of the classical model \eqref{2ilw}. More precisely, we show that the usual structures of integrability, listed above for the standard ILW equation, are present also in \eqref{2ilw}.

\subsection{Basic properties of the non-chiral ILW equation}
\label{subsec1.1}
For the convenience of the reader, we shortly discuss some basic properties of the system in \eqref{2ilw}; see \cite{berntson2020a} for further details. 

For the physical interpretation of \eqref{2ilw}, it is useful to represent the operators in \eqref{TT} in Fourier space \cite[Eq.\ (A3)]{berntson2020a}:  
\begin{equation*}
\begin{split} 
(\widehat{Tu})(k) &=\mathrm{i}\coth(k\delta)\hat{u}(k),\\ 
(\widehat{\tilde{T}u})(k) &=\mathrm{i}\,\csch(k\delta)\hat{u}(k),
\end{split} 
\end{equation*}
where $\csch(z)\coloneqq 1/\sinh(z)$. Thus, the dispersion relations corresponding to the $T$- and $\tilde{T}$-terms in \eqref{2ilw} are  $\pm k^2\coth(k\delta)$ and $\pm k^2\csch(k\delta)$, respectively. 
This makes manifest that, in the limit $\delta\to 0$, $T$ converges to the Hilbert transform and  $\tilde{T}$ to zero, i.e., in this limit, \eqref{2ilw} reduces to two uncoupled Benjamin-Ono (BO) equations (see \cite{chen1979}  and references therein for  background on the BO equation). Since the BO equation for $u(x,t)$ in this limit is chiral in the sense that it can only describe solitons moving to the right \cite{chen1979}, the corresponding BO equation for $v(x,t)$ is also chiral in that it only describes left-moving solitons (to see the latter, note that the BO equation for $v$ is obtained from the one for $u$ the by the transformation $v(x,t)=u(-x,t)$). 
Thus, one should expect that, for finite $\delta$, the system  \eqref{2ilw} is non-chiral in the sense that it can describe solitons in both directions and where solitons that move in opposite directions can interact in a non-trivial way; this expectation is confirmed by the multi-soliton solutions of \eqref{2ilw} obtained in  \cite[Section VI A]{berntson2020a}.  

It is known that the ILW equation in the limit $\delta\to 0$ reduces to the Korteweg-de Vries (KdV) equation 
\cite{kodama1981}. Moreover, systematic classifications of  integrable two-component systems of KdV-type exist in the literature \cite{foursov2003,karasu1997}. 
This suggests that our system in \eqref{2ilw} might have a known KdV-limit $\delta\to 0$. 
However, at closer inspection, one finds that this is not the case: the system in \eqref{2ilw} has no well-defined non-trivial limit $\delta\to 0$ (the interested reader can find details on this in Appendix~\ref{app:KdV}).
Thus, the interesting case for \eqref{2ilw} is where $\delta$ is in the range $0<\delta<\infty$. 

It is interesting to note that the system in \eqref{2ilw} allows for a reduction $v(x,t)=u(-x,t)$; this yields 
\begin{equation*} 
u_t(x,t) + 2 u(x,t) u_x(x,t) + Tu_{xx}(x,t)+\tilde{T}u_{xx}(-x,t)=0, 
\end{equation*} 
which is a non-local extension of the ILW equation where the temporal change of $u$ at position $x$ is affected by $\tilde{T}u_{xx}$ at position $-x$. 
We note in passing that integrable equations with nonlocalities in the form $u(-x,t)$ have recently attracted considerable attention \cite{ablowitz2017}. 
However, our system in \eqref{2ilw}, without such a non-locality, allows for a more conventional physical interpretation.  

To set our results in perspective, it is interesting to note that there is also a system superficially similar to \eqref{2ilw} but differing by signs which, from a physics point of view, make an important difference: 
\begin{equation*} 
\label{2ilwchiral} 
\begin{split} 
&u_t + 2 u u_x + Tu_{xx}+\tilde{T}v_{xx}=0,\\
&v_t + 2 v v_x + Tv_{xx}+\tilde{T}u_{xx}=0. 
\end{split} 
\end{equation*} 
One can show that this system is integrable \cite{berntson2021}; it is chiral in the sense that it only allows for solitons moving to the right, and it allows for a local reduction $v(x,t)=u(x,t)$: 
\begin{equation*} 
u_t + 2 uu_x + Tu_{xx}+\tilde{T}u_{xx}=0 ; 
\end{equation*} 
however, the latter is just the standard ILW equation with $\delta$ replaced by $\delta/2$ and thus not new (this follows from $T_\delta+\tilde{T}_\delta=T_{\frac{\delta}2}$, which can be easily checked using the definitions). 

From the point of view of applications to physics, it is natural to restrict $u$ and $v$ in the non-chiral ILW equation \eqref{2ilw} to be real-valued functions; 
however, it turns out that the results in this paper extend naturally to complex-valued functions, and we therefore allow $u$ and $v$ to be complex-valued. 

It is interesting to note that the non-chiral ILW equation in \eqref{2ilw} has multi-soliton solutions such that the time evolution of the poles is governed by the dynamics of the hyperbolic $A$-type Calogero-Moser systems  \cite[Section IV A]{berntson2020a}, in natural generalization of a famous result for the BO equation \cite{chen1979}.  Moreover, these solutions can be generalized to a periodic variant of  \eqref{2ilw}; in this case, the pole dynamics is governed by elliptic $A$-type Calogero-Moser systems \cite[Section IV A]{berntson2020a}. As we plan to report elsewhere, the results in the present paper can be generalized to the periodic case.  

\subsection{Possible physics applications}
\label{subsec1.2} 
As already mentioned, we first discovered a quantum version of the non-chiral ILW equation in the context of the fractional quantum Hall effect \cite[Section III]{berntson2020a}. 
This quantum version of equation \eqref{2ilw} arises from an underlying non-chiral conformal field theory  \cite[Section III]{berntson2020a}, different from the conformal field theory giving raise to a quantum version of the BO equation which is chiral \cite{abanov2009}.

Soliton equations are known for being widely applicable \cite{calogero2001}, and we therefore believe that the non-chiral ILW equation will also find applications in other areas of physics. 
To elaborate on one possible such example, we recall that the fundamental equations describing nonlinear water waves are invariant under the parity transformation $x\to -x$, and real solitary waves can move in both directions, left and right. 
However, known soliton equations such as the KdV, BO, and (standard) ILW equations can only describe solitons moving in one direction. One can, of course, extend such a chiral soliton equation by adding a corresponding equation for $v(x,t)=u(-x,t)$, but this system is trivial in the sense that solitons moving in opposite directions do not interact; this is not realistic from a physics point of view. Thus, one can regard \eqref{2ilw} as an integrable extension of such a trivial system where solitons moving in opposite directions interact in a particular way so as to preserve integrability.

The above discussion together with arguments by Calogero \cite{calogero2001} suggest that it would be worthwhile to revisit the derivation of the standard ILW equation from more fundamental parity invariant equations describing nonlinear water waves  \cite{joseph1977}, and thus try to justify \eqref{2ilw} as a description of nonlinear water waves including interaction effects that were missed in previously known soliton equations.

A more detailed discussion of possible applications of the non-chiral ILW equation in other areas of physics can be found in \cite[Section VI]{berntson2020a}.

\subsection{Related work}
\label{subsec1.3} 
We briefly mention some previous work on related integro-differential equations.
An integrable extension of the nonlinear Schr\"odinger equation with a nonlocal term involving the integral operator $T$ in \eqref{TT} was proposed in \cite{pelinovsky1995},  
 and the Hirota form of this so-called  intermediate nonlinear Schr\"odinger (INLS) equation was obtained in \cite{pelinovsky1995,matsuno2000}. 
 It is interesting to note that the Hirota form of INLS equation is similar to the Hirota form of the non-chiral ILW equation in \eqref{hirota_form}, but they are not the same; in particular, the latter includes shifts (see \eqref{shifts}) and the former does not but, instead, involves complex conjugation (see \cite[Eqs.\  (15a--c)]{pelinovsky1995} or \cite[Eqs.\ (8)--(9)]{matsuno2000}). 
The Lax pair and conservation laws for the INLS equation were obtained in \cite{pelinovskygrimshaw1995} and explored in \cite{matsuno2002,matsuno2004}.  
We also mention Ref.~\cite{zhang1994} 
 where a class of integrable integro-differential equations of ILW-type was studied, including a system combining an ILW equation and (linear) Schr\"{o}dinger equations in a nontrivial way; it would be interesting to try to modify this approach to accommodate the non-chiral ILW equation. 
 
\subsection{Plan}
\label{subsect1.4}
The plan of this paper is as follows. In Section \ref{laxsec}, we derive a Lax pair for \eqref{2ilw}. A Hirota bilinear form is presented in Section \ref{hirotasec}, where we additionally prove that the Hirota bilinear form is equivalent to \eqref{2ilw} by constructing explicit transformations from $(u,v)$ to the Hirota variables $(F,G)$ and vice-versa.  A B\"acklund transformation is constructed from the Hirota bilinear form in Section \ref{backlundsec}. Section \ref{conservationsec} contains two independent, complementary derivations of an infinite sequence of conservation laws.  
 Some properties of the operators $T$ and $\tilde{T}$ defined in \eqref{TT} are collected in Appendix~\ref{app:TT}, and details on the KdV-limit $\delta\to 0$ can be found in Appendix~\ref{app:KdV}.

In what follows we assume that the arguments of $T$ and $\tilde{T}$ are sufficiently regular and decay sufficiently rapidly to justify our arguments. We occasionally comment on specific necessary or sufficient conditions for clarity.

\section{Lax pair}\label{laxsec}

In this section we derive a Lax pair for \eqref{2ilw}. The ansatz used to obtain the Lax pair is a generalization of the ansatz used in \cite{ablowitz1982} to find a Lax pair for the standard ILW equation; the operative difference here is that the underlying Riemann-Hilbert (RH) problem on the cylinder has a pair of jumps instead of a single jump.  In Section \ref{conservationsec}, we will illustrate the utility of the Lax pair by using it to derive an infinite number of conservation laws for \eqref{2ilw}.

For fixed $\delta > 0$, let $C$ denote the cylinder $C=\C/2\ii\delta\Z$ and let $\pi:\C\to C$ be the natural projection. 
We can identify $C$ with the strip
\begin{equation*}
C\simeq \{z\in\C: 0 \leq \Im z < 2\delta \}
\end{equation*}
and a function $f:C\to\C$ can be viewed as a function $f:\C\to\C$ which is periodic with period $2\ii\delta$, i.e.,
\begin{equation*}
f(z+2\ii n\delta)=f(z),\qquad n\in\Z.
\end{equation*}
Let $C_0$ and $C_\delta$ denote the images of the lines $\Im z=0$ and $\Im z=\delta$, respectively, under $\pi$. We consider an eigenfunction, $\psi(z,t;k)$,  that will appear in the Lax pair; 
for each $t\in\R$ and $k\in\C$, $\psi(z)=\psi(z,t;k)$ is an analytic function $C\setminus(C_0\cup C_\delta)\to\C$ with jumps across $C_0$ and $C_\delta$. The boundary values of the eigenfunction, $\psi^{\pm}(z)$, are functions $C_0\cup C_\delta\to\C$ defined by
\begin{align}\label{psi_bv}
\psi^{\pm}(x,t;k)\coloneqq\lim\limits_{\epsilon \downarrow 0} \psi(x\pm\ii\epsilon,t;k),\qquad \psi^{\pm}(x+\ii\delta,t;k)\coloneqq\lim\limits_{\epsilon \downarrow 0} \psi(x+\ii\delta\pm\ii\epsilon,t;k).
\end{align}
We consider the following ansatz for a Lax pair for \eqref{2ilw}:
\begin{equation}\label{lax_ansatz}
\begin{cases}
\ii \psi^-_x+(-u-\mu_1)\psi^-=\nu_1\psi^+, \qquad &\text{for } z\in C_0, \\
\ii \psi^+_x+(v-\mu_2)\psi^+=\nu_2\psi^-, \qquad &\text{for } z\in C_\delta, \\
\psi_t+\ii \psi_{xx}-\ii A(z,t;k)\psi-\ii B(z,t;k)\psi_x=0, \qquad & \text{for } z\in C\setminus(C_0\cup C_\delta),
\end{cases}
\end{equation}
where $\mu_1=\mu_1(k)$, $\mu_2=\mu_2(k)$, $\nu_1=\nu_1(k)$, and $\nu_2=\nu_2(k)$ are complex-valued functions of the spectral parameter $k\in\C$, and $A(z,t;k)$ and $B(z,t;k)$ are, for each $t\in \R$ and $k\in \C$,  bounded analytic functions $\C\setminus(C_0\cup C_\delta)\to\C$ to be determined. 

To obtain the compatibility conditions for \eqref{lax_ansatz}, we write the boundary values of the $t$-part of the Lax pair:
\begin{equation*}
\psi^{\pm}_t+\ii \psi^{\pm}_{xx}-\ii A^{\pm}(z,t;k)\psi-\ii B^{\pm}(z,t;k)\psi^{\pm}_x=0, \qquad \text{for } z\in C_0\cup C_\delta.
\end{equation*}
This equation and its $x$-derivative can be used to eliminate $\psi^{\pm}_t$ and $\psi^{\pm}_{tx}$ from the $t$-derivative of the $x$-part of \eqref{lax_ansatz}, leading to 
\begin{align*}
& \ii \nu_1(B^+ - B^-) \psi_x^+ 
+ \ii\nu_1(A^- - A^+ + B_x^- +  2\ii u_x) \psi^+
	\\
&+ \big[-u_t - A_x^- + \ii (\mu_1 + u)B_x^- - 2\mu_1 u_x + \ii B^- u_x - 2uu_x - \ii u_{xx}\big]\psi^+ =  0, \qquad \text{on $C_0$},
\end{align*}
and
\begin{align*}
& \ii \nu_2(B^+ - B^-) \psi_x^+ 
+ \ii\nu_2(A^+ - A^- + B_x^+ - 2\ii v_x) \psi^+
	\\
&+ \big[v_t - A_x^+ + \ii (\mu_2 - v)B_x^+ + 2\mu_2 v_x - \ii B^+ v_x - 2vv_x + \ii v_{xx}\big]\psi^+ =  0, \qquad \text{on $C_\delta$}.
\end{align*}
Setting the coefficients of $\psi_x^\pm$ and $\psi^\pm$ to zero, we find the equations
\begin{subequations}\label{bicond}
\begin{align}\label{biconda2}
& B^+ - B^- = 0, \qquad \text{on $C_0 \cup C_\delta$},
	\\\label{bicondb2}
& A^+ -A^- - B_x^- -  2\ii u_x = 0, \qquad \text{on $C_0$},
	\\\label{bicondc2}
& A^+ - A^- + B_x^+ - 2\ii v_x = 0, \qquad \text{on $C_\delta$},
 	\\\label{biconde2}
& u_t + A_x^- - \ii(\mu_1 + u)B_x^- + 2\mu_1 u_x - \ii B^- u_x + 2uu_x + \ii u_{xx} = 0,\qquad \text{on $C_0$},
	\\\label{bicondd2}
& v_t - A_x^+ + \ii(\mu_2 -v)B_x^+ + 2\mu_2 v_x - \ii B^+ v_x - 2vv_x + \ii v_{xx} =  0, \qquad \text{on $C_\delta$}.
\end{align}
\end{subequations}
Condition \eqref{biconda2} shows that $B$ has no jump across $C_0 \cup C_\delta$. Hence $B$ is a bounded analytic function $C \to \C$, so $B(z,t) =B_0$ must be a constant. Condition \eqref{bicondb2} then shows that $A$ is a solution of the following scalar RH problem on $C$:
\begin{itemize}
\item $A:C \setminus (C_0 \cup C_\delta) \to \C$ is an analytic function, 
\item across $C_0 \cup C_\delta$, $A$ satisfies the jump condition
$$A^+(z) - A^-(z) = \begin{cases} 2\ii u_x(x),  &z = x \in C_0, \\
2\ii v_x(x),  &z = x+ \ii\delta \in C_\delta, 
\end{cases}$$
\item $A(z) = \mathrm{O}(1)$ as $z \to \infty$.
\end{itemize}

To solve this problem we use the following lemma. 

\begin{lemma}[RH problem on $C$ with a jump across $C_0 \cup C_\delta$]\label{RHlemma}
Let $J_0: C_0 \to \C$ and $J_1: C_\delta \to \C$ be continuous functions such that
$$\int_\R J_0(x)\, \mathrm{d}x = \int_\R J_1(x)\,\mathrm{d}x = 0.$$ 
Define $J: C_0 \cup C_\delta \to \C$ by
$$J(z) \coloneqq \begin{cases} J_0(x), & z \in C_0, \\
J_1(x), & z \in C_\delta.
\end{cases}$$
Then the scalar RH problem
\begin{itemize}
\item $A:C \setminus (C_0 \cup C_\delta)$ is analytic, 
\item across $C_0 \cup C_\delta$, $A$ satisfies the jump condition
$$A^+(z) - A^-(z) = J(z), \qquad z \in C_0 \cup C_\delta,$$
\item $A(z) = \mathrm{O}(z^{-1})$ as $z \in C$ approaches infinity, 
\end{itemize}
has the unique solution
\begin{align}\label{rh_soln}
A(z) = \frac{1}{4\delta \ii} \int_{C_0 \cup C_\delta} \coth\bigg(\frac{\pi(z'-z)}{2\delta}\bigg) J(z')\, \mathrm{d}z', \qquad z \in C \setminus (C_0 \cup C_\delta),
\end{align}
where both $C_0$ and $C_\delta$ are oriented from $\re\, z=-\infty$ to $\re\,z=\infty$.

Moreover, this solution satisfies
$$A^\pm(z) = \begin{cases}
\frac{(TJ_0)(x) + (\tilde{T}J_1)(x)}{2\ii} \pm \frac{1}{2} J_0(x), & z = x \in C_0 \cong \R,
	\\
\frac{(\tilde{T}J_0)(x) + (TJ_1)(x)}{2\ii} \pm \frac{1}{2} J_1(x), & z = x + \ii \delta \in C_\delta \cong \R + \ii \delta,
\end{cases}
$$
with the operators $T$ and $\tilde{T}$ as in \eqref{TT}. 
\end{lemma}
\begin{proof}
If $A_1$ and $A_2$ are two different solutions, then $A_1 - A_2$  is analytic on $C$ and of order $\mathrm{O}(z^{-1})$ as $z \in C$. Hence, by Liouville's theorem, $A_1 = A_2$. This proves uniqueness. 

Let $A$ be given by \eqref{rh_soln}. For all $z \in \C$ and $n \in \Z$, we have
$$\coth\bigg(\frac{\pi(z + 2\ii n\delta )}{2\delta}\bigg) = \coth\bigg(\frac{\pi z}{2\delta}\bigg).$$
Hence $A$ descends to an analytic function $C \setminus (C_0 \cup C_\delta) \to \C$. 
For $x \in C_0 \cong \R$, the Plemelj formula gives
\begin{align*}
A^\pm(x)  
= &\;\frac{1}{4\delta \ii} \bigg\{\pvint_{\R} \coth\bigg(\frac{\pi(z'-x)}{2\delta}\bigg) J(z')\,\mathrm{d}z'
+ \int_{\R + \ii\delta} \coth\bigg(\frac{\pi(z'-x)}{2\delta}\bigg) J(z')\,\mathrm{d}z'
	\\
& \pm \ii\pi\,\underset{z' = x}\Res\coth\bigg(\frac{\pi(z'-x)}{2\delta}\bigg) J_0(z')\bigg\}
	\\
= &\; \frac{1}{2\ii} (TJ_0)(x) 
+ \frac{1}{4\delta \ii} \int_{\R} \coth\bigg(\frac{\pi(x'+\ii\delta-x)}{2\delta}\bigg) J_1(x')\, \mathrm{d}x'
\pm \frac{1}{2} J_0(x)
	\\
= &\; \frac{1}{2\ii} (TJ_0)(x) + \frac{1}{2\ii} (\tilde{T}J_1)(x) \pm \frac{1}{2} J_0(x).
\end{align*}

Similarly, for $x + \ii\delta \in C_\delta \cong \R + \ii\delta$,
\begin{align*}
A^\pm(x + \ii\delta) 
= &\; \frac{1}{4\delta \ii} \bigg\{\int_{\R} \coth\bigg(\frac{\pi(z'-x-\ii\delta)}{2\delta}\bigg) J(z') \,\mathrm{d}z'
	\\
& + \pvint_{\R + \ii\delta} \coth\bigg(\frac{\pi(z'-x-\ii\delta)}{2\delta}\bigg) J(z') \,\mathrm{d}z'
	\\
& \pm \ii\pi  \underset{z' = x+\ii\delta}\Res \coth\bigg(\frac{\pi(z'-x-\ii\delta)}{2\delta}\bigg) J(z')\bigg\}
	\\
= &\; \frac{1}{2\ii} (\tilde{T}J_0)(x) + \frac{1}{2\ii} (TJ_1)(x) \pm \frac{1}{2} J_1(x).
\end{align*}
This proves the expressions for the boundary values and shows that $A$ satisfies the correct jump condition. 

As $x \to \pm\infty$, we have $\coth(x+\ii y) = 1 + 2e^{\mp 2\ii y} e^{-2|x|} + \mathrm{O}(e^{-4|x|})$ uniformly for $y\in\R$. Hence the assumption that $\int_\R J_0(x)\,\mathrm{d}x = 0$ and $\int_\R J_1(x)\,\mathrm{d}x = 0$ implies that $A(z)\to0$ as $z \in C$ tends to $\infty$.
\end{proof}

Using \ref{RHlemma}, we find
\begin{align}\label{A_rh_soln}
A(z,t;k)=&\; \frac{1}{2\delta}\int_\R \coth\bigg(\frac{\pi(x'-z)}{2\delta}\bigg) u_x(x')\,\mathrm{d}x' \\
&+\frac{1}{2\delta}\int_\R \coth\bigg(\frac{\pi(x'+\ii\delta-z)}{2\delta}\bigg) v_x(x')\,\mathrm{d}x'+ A_0(k), \qquad z\in C\setminus (C_0\cup C_\delta), \nonumber
\end{align}
and 
\begin{align*}
A^\pm(z,t;k) = \begin{cases} 
(Tu_x)(x)  + (\tilde{T}v_x)(x) \pm \ii u_x(x) + A_0(k), & z = x \in C_0,
	\\
(Tv_x)(x)  + (\tilde{T}u_x)(x) \pm \ii v_x(x)+ A_0(k), & z = x + \ii\delta \in C_\delta. 
\end{cases}
\end{align*}
Substituting these expressions for $A^\pm$ into \eqref{bicondd2} and \eqref{biconde2} and using that $T$ and $\tilde{T}$ commute with $\partial_x$ from Proposition \ref{TpropertiesR}, we arrive at the two-component equation
\begin{align*}
& u_t + Tu_{xx} + \tilde{T}v_{xx} + 2\mu_1 u_x - \ii B_0 u_x + 2uu_x  = 0,
	\\
& v_t - Tv_{xx} - \tilde{T}u_{xx} + 2\mu_2 v_x - \ii B_0 v_x - 2vv_x =  0.
\end{align*}
Choosing $\mu_1 = \mu_2 = \mu$ and $B_0 = -2\ii\mu$, this becomes the non-chiral ILW equation \eqref{2ilw}. 
We summarize the results above in a theorem.

\begin{theorem}[Lax pair for the non-chiral ILW equation]
The non-chiral ILW equation \eqref{2ilw} is the compatibility condition of the Lax pair
\begin{align}\label{2ilw_laxpair}
\begin{cases}
\ii\psi_x^- + (-u-\mu)\psi^- = \nu_1 \psi^+, & \text{on $C_0$},
	\\
\ii\psi_x^+ + (v-\mu)\psi^+ = \nu_2 \psi^-, & \text{on $C_\delta$},
	\\
\psi_t^\pm + i \psi_{xx}^\pm - 2\mu \psi_x^\pm - \ii (Tu_x + \tilde{T}v_x \pm \ii u_x + A_0) \psi^\pm=0, & \text{on $C_0$},
	\\
\psi_t^\pm + i \psi_{xx}^\pm - 2\mu \psi_x^\pm - \ii (Tv_x  + \tilde{T}u_x  \pm \ii v_x + A_0) \psi^\pm=0, & \text{on $C_\delta$},
\end{cases}
\end{align}
where $\mu=\mu(k)$, $\nu_1=\nu_1(k)$, $\nu_2=\nu_2(k)$, and $A_0=A_0(k)$ are complex parameters which may depend on the spectral parameter $k$. 
\end{theorem}
\begin{remark}\upshape
The $t$-parts of \eqref{2ilw_laxpair} have an analytic extension to $C \setminus (C_0 \cup C_\delta)$ and can be alternatively written as
$$\psi_t + \ii \psi_{xx} - 2\mu \psi_x - \ii A \psi = 0, \qquad \ z\in C \setminus (C_0 \cup C_\delta),$$
where $A=A(z,t;k)$ is given by \eqref{A_rh_soln}. 
\end{remark}

\section{Hirota bilinear form}\label{hirotasec}
In this section, we show that the bilinear system
\begin{subequations}\label{hirota_form}
\begin{align}
&(\ii D_t-D_x^2) F^-\cdot G^+=0, \label{hirota_form_a}\\
&(\ii D_t-D_x^2)F^+\cdot G^-=0, \label{hirota_form_b}
\end{align}
\end{subequations}
where $D_t$ and $D_x$ are the usual Hirota derivatives:
\begin{equation}
D_t^mD_x^n F\cdot G \coloneqq (\partial_t-\partial_{t'})^m(\partial_x-\partial_{x'})^n FG|_{t'=t,x'=x},
\end{equation}
and
\begin{equation}\label{shifts}
F^{\pm}(x,t)=F(x\pm\ii\delta/2,t), \quad G^{\pm}(x,t)=G(x\pm\ii\delta/2,t),
\end{equation}
is equivalent to \eqref{2ilw} in the sense of the following theorem.

\begin{theorem}[Hirota bilinear form of the non-chiral ILW equation]\label{bilinearth}
\hfill \break\nopagebreak 
\begin{roster}
\item
Let $F(z,t)$ and $G(z,t)$ be functions of $z\in\C$ and $t\in  \R$ such that $\log F(z,t)$ and $\log G(z,t)$ are analytic for $-\delta/2 < \im\, z<\delta/2$ and continuous for $-\delta/2 \leq \im\, z \leq \delta/2$, and
\begin{align}\label{FGlimits}
\begin{cases}
\log F(x+\ii y,t)=f_0\pm f_1 x+\mathrm{O}(x^{-1}) \\ \log G(x+\ii y,t)=g_0\pm f_1 x+\mathrm{O}(x^{-1})
\end{cases} \text{ as } x\to\pm\infty, \qquad -\frac{\delta}{2}\leq y\leq \frac{\delta}{2},
\end{align}
for some constants $f_0,g_0,f_1\in \C$. Then $F,G$ satisfy the bilinear system \eqref{hirota_form} if and only if the functions
\begin{equation}\label{FG_to_uv}
u=\ii \partial_x\log \frac{F^-}{G^+},\qquad v=\ii\partial_x\log \frac{G^-}{F^+},
\end{equation}
satisfy \eqref{2ilw}. 
\\
\item 
Suppose $u(x,t)$, $v(x,t)$ are solutions of \eqref{2ilw} and the transforms $Tu$, $Tv$, $\tilde{T}u$, $\tilde{T}v$ exist. Then $F(x,t)$, $G(x,t)$, defined up to inessential multiplicative constants by 
\begin{equation}
\begin{split}\label{FG_from_uv}
\begin{dcases}
\ii \partial_z \log F(z,t)=\frac{1}{2\delta\ii} \int_\R \tanh\bigg(\frac{\pi(x'-z)}{\delta}\bigg)\big(u_{+}(x')+v_{+}(x')\big)\,\mathrm{d}x', \\ 
\ii\partial_z \log G(z,t)=\frac{1}{2\delta\ii}\int_\R \tanh\bigg(\frac{\pi(x'-z)}{\delta}\bigg)\big(u_{-}(x')+v_{-}(x')\big)\,\mathrm{d}x',
\end{dcases}
\end{split}
\end{equation}
where
\begin{equation}
\begin{split}\label{uvprojections}
\begin{cases}
u_{\pm}\coloneqq\frac12 u\mp\frac{\ii}{2}\big(Tu+\tilde{T}v \big)\\
v_{\pm}\coloneqq\frac12 v\pm\frac{\ii}{2}\big(Tv+\tilde{T}u\big)
\end{cases}\qquad x,t\in \R,
\end{split}
\end{equation}
and \eqref{FGlimits}, are analytic for $-\delta/2<\im\,z<\delta/2$, continuous for $-\delta/2\leq\im\, z\leq \delta/2$, and satisfy the Hirota equations \eqref{hirota_form}.
\end{roster}
\end{theorem}

\subsection{Proof of Theorem \ref{bilinearth}A}
Suppose $(F,G)$ and $(u,v)$ are related as in \eqref{FG_to_uv}. We write 
\begin{equation*}
u(x,t)=u_+(x,t)+u_-(x,t),\qquad v(x,t)=v_+(x,t)+v_-(x,t),
\end{equation*}
where $u_{\pm}$ and $v_{\pm}$ are defined by
\begin{equation}\label{uvplusminusdef}
\begin{split}
&u_+(z,t)\coloneqq\ii\partial_z\log F(z-\ii\delta/2,t),\qquad u_-(z,t)\coloneqq-\ii\partial_z\log G(z+\ii\delta/2,t), \\
&v_+(z,t)\coloneqq-\ii\partial_z\log F(z+\ii\delta/2,t),\qquad v_-(z,t)\coloneqq\ii\partial_z\log G(z-\ii\delta/2,t),
\end{split}
\end{equation}
with $z$ the complex extension of $x$. By our assumptions on the analyticity of $\log F$ and $\log G$, we see that $u_+$ and $v_-$ are analytic in the strip $0< \im z<\delta$ and $u_-$ and $v_+$ are analytic in the strip $-\delta<\im z< 0$. Additionally, we observe that
\begin{equation}\label{uvplusminusrelations}
v_+(z,t)=-u_+(z+\ii\delta,t),\qquad v_-(z,t)=-u_-(z-\ii\delta,t).
\end{equation}

\begin{lemma}\label{lemma1}
If $g_+(z)$ is analytic in the strip $0 < \im z < \delta$, continuous in the strip $0 \leq \im z \leq \delta$, and 
\begin{align}\label{limgplus}
\lim_{R \to \infty} \bigg(\int_{-R}^{-R+ \ii\delta} + \int_{R}^{R+ \ii\delta}\bigg) g_+(z)\,\mathrm{d}z = 0,
\end{align}
then
\begin{align}\label{TTtildegplus}
(Tg_+)(x) - (\tilde{T}[g_+(\cdot + \ii\delta)])(x) = \ii g_+(x), \qquad x \in \R.
\end{align}
Similarly, if $g_-(z)$ is analytic in the strip $-\delta < \im z < 0$, continuous in the strip $-\delta \leq \im z \leq 0$, and 
\begin{align}\label{limgminus}
\lim_{R \to \infty} \bigg(\int_{-R-\ii\delta}^{-R} + \int_{R- \ii\delta}^{R}\bigg) g_-(z)\,\mathrm{d}z = 0,
\end{align}
then
\begin{align}\label{TTtildegminus}
(Tg_-)(x) - (\tilde{T}[g_-(\cdot - \ii\delta)])(x) = -\ii g_-(x), \qquad x \in \R.
\end{align}
\end{lemma}
\begin{proof}
Suppose $g_+(z)$ is a function which is analytic in $0 < \im z < \delta$, continuous in $0 \leq \im z \leq \delta$, and which satisfies \eqref{limgplus}. Using the definition of $\tilde{T}$ and then changing variables to $z' = x' +\ii\delta$, we find
\begin{align*}
(\tilde{T}[g_+(\cdot + \ii\delta)])(x) 
& = \frac{1}{2\delta} \int_\R \tanh\bigg(\frac{\pi(x'-x)}{2\delta}\bigg) g_+(x' + \ii\delta)\,\mathrm{d}x'
  	\\
& = \frac{1}{2\delta} \int_{\R + \ii\delta} \tanh\bigg(\frac{\pi(z'-x)}{2\delta} + \frac{\ii\pi }{2}\bigg) g_+(z')\,\mathrm{d}z'.
\end{align*}
We next use the identity $\tanh(z + \ii\pi/2) = \coth(z)$ and deform the contour down towards the real axis. Utilizing the Plemelj formula to evaluate the contribution from the simple pole at $z' = x$, we obtain
\begin{align}\nonumber
(\tilde{T}[g_+(\cdot + \ii\delta)])(x) 
= &\; \frac{1}{2\delta} \pvint_{\R} \tanh\bigg(\frac{\pi(z'-x)}{2\delta} + \frac{\ii\pi }{2}\bigg) g_+(z')\, \mathrm{d}z'
	\\ \label{tildeTgpluscdot}
& - \frac{\pi \ii}{2\delta}\, \underset{z'=x}\res \coth\bigg(\frac{\pi(z'-x)}{2\delta}\bigg) g_+(z')
+ E(x),
\end{align}
where $E(x)$ is defined by
$$E(x)\coloneqq \frac{1}{2\delta}  \lim_{R \to \infty} 
\bigg(\int_{-R+ \ii\delta}^{-R} + \int_{R}^{R+ \ii\delta} \bigg)\coth\bigg(\frac{\pi(z'-x)}{2\delta}\bigg) g_+(z') \,\mathrm{d}z'.$$
As $R \to \infty$, the function $\coth(\frac{\pi(z'-x)}{2\delta})$ tends to $1$ (resp. $-1$) uniformly for $z' \in [R, R+\ii\delta]$ (resp. $z' \in [-R, -R+\ii\delta]$). Hence, thanks to the assumption \eqref{limgplus}, we have $E = 0$.
Thus, using that 
$$\underset{z' = x}{\res} \coth\bigg(\frac{\pi(z'-x)}{2\delta}\bigg) = \frac{2\delta}{\pi},$$
equation \eqref{tildeTgpluscdot} reduces to 
\begin{align}\nonumber
(\tilde{T}[g_+(\cdot + \ii\delta)])(x) = (T g_+)(x) - \ii g_+(x'),
\end{align}
which is \eqref{TTtildegplus}. The proof of \eqref{TTtildegminus} is similar.
\end{proof}

\begin{lemma}\label{lemma2}
The functions $u$ and $v$ obey the identities
$$  \begin{cases}
  Tu + \tilde{T}v = \ii(u_+ - u_-), \\
  Tv + \tilde{T}u = -\ii(v_+ - v_-), 
\end{cases} \qquad x,t \in \R.$$
\end{lemma}
\begin{proof}
By \eqref{uvplusminusrelations}, we have
\begin{align}\nonumber
Tu + \tilde{T}v = &\; T(u_+ + u_-) + \tilde{T}(v_+ + v_-) \\
= &\; Tu_+ - \tilde{T}[u_+(\cdot + \ii\delta)] + Tu_- - \tilde{T}[u_-(\cdot - \ii\delta)]. \label{TutildeTv}
\end{align}
We see from \eqref{uvplusminusdef} that $u_+$ is analytic for $0 < \im z < \delta$, continuous for $0 \leq \im z \leq \delta$, and satisfies \eqref{limgplus} by virtue of \eqref{FGlimits}. Similarly, $u_-$ is analytic for $-\delta < \im z < 0$, continuous for $-\delta \leq \im z \leq 0$, and satisfies \eqref{limgminus}. Hence, the identity $Tu + \tilde{T}v = \ii(u_+ - u_-)$ follows from \eqref{TutildeTv} and Lemma \ref{lemma1}. 

The proof of the identity $Tv + \tilde{T}u = -\ii(v_+ - v_-)$ is similar. Indeed, by \eqref{uvplusminusrelations}, we have
\begin{align*}
Tv + \tilde{T}u = T(v_+ + v_-) + \tilde{T}(u_+ + u_-)
= Tv_+ - \tilde{T}[v_+(\cdot - \ii\delta)] + Tv_- - \tilde{T}[v_-(\cdot + \ii\delta)].
\end{align*}
We see from \eqref{uvplusminusdef} that $v_-$ is analytic for $0 < \im z < \delta$, continuous for $0 \leq \im z \leq \delta$, and satisfies \eqref{limgplus} by virtue of \eqref{FGlimits}. Similarly, $v_+$ is analytic for $-\delta < \im z < 0$, continuous for $-\delta \leq \im z \leq 0$, and satisfies \eqref{limgminus}.
Hence, the identity $Tv + \tilde{T}u = -\ii(v_+ - v_-)$ follows from Lemma \ref{lemma1}.
\end{proof}

We are now in a position to prove Theorem \ref{bilinearth}A.

\begin{proof}[Proof of Theorem \ref{bilinearth}A]
According to Lemma \ref{lemma2}, the non-chiral ILW equation \eqref{2ilw} can be written as
\begin{align}\label{2ILW2}
\begin{cases}
u_t  + 2uu_x + \ii(u_+ - u_-)_{xx} =  0,
	\\
v_t  - 2vv_x + \ii(v_+ - v_-)_{xx} =  0.
\end{cases}
\end{align}
Since $u_t =\ii\big(\log \frac{F^-}{{G}^+}\big)_{xt}$ and $v_t = \ii\big(\log \frac{G^-}{{F}^+}\big)_{xt}$, integration of \eqref{2ILW2} with respect to $x$ gives
\begin{align*}
\begin{cases}
\ii\big(\log \frac{F^-}{G^+}\big)_{t}  + u^2 + \ii(u_+ - u_-)_{x} = \lambda_1(t),
	\\
\ii\big(\log \frac{G^-}{{F}^+} \big)_{t}  - v^2 + \ii(v_+ - v_-)_{x} =-\lambda_2(t),
\end{cases}
\end{align*}
where $\lambda_1(t)$ and $\lambda_2(t)$ are complex-valued functions. By considering the limit $x \to \infty$ and using \eqref{FGlimits} and \eqref{uvplusminusdef}, we conclude that $\lambda_1(t) = \lambda_2(t) = 0$. 
Rewriting the system in terms of $F$ and $G$, we obtain
\begin{align*}
\begin{cases}
\ii\big(\frac{F^-_t}{F^-} - \frac{{G}^+_t}{{G}^+}\big) - \big(\frac{F^-_x}{F^-} - \frac{{G}^+_x}{{G}^+}\big) ^2 - \big(\frac{F^-_x}{F^-} + \frac{{G}^+_x}{{G}^+}\big)_{x} = 0,
	\\
-\ii\big(\frac{{F}^+_t}{{F}^+} - \frac{G^-_t}{F^-} \big) 
+ \big(\frac{G^-_x}{G^-}- \frac{{F}^+_x}{{F}^+}\big)^2 
+ \big(\frac{{F}^+_x}{{F}^+}+ \frac{G^-_x}{G^-}\big)_{x} = 0.
\end{cases}
\end{align*}
Simplification shows that the first equation can be rewritten as
\begin{align*}
\ii\bigg(\frac{F^-_t}{F^-} - \frac{{G}^+_t}{{G}^+}\bigg) - \frac{F^-_{xx}}{F^-} + \frac{2F^-_x{G}^+_x}{F^-{G}^+} - \frac{{G}^+_{xx}}{{G}^+} = 0,
\end{align*}
i.e.,
\begin{subequations}\label{2ILW3}
\begin{align}\label{2ILW3a}
\frac{(\ii D_t - D_x^2)F^- \cdot {G}^+}{F^-{G}^+} = 0.
\end{align}
In the same way, the second equation can be written as
\begin{align}\label{2ILW3b}
\frac{(-\ii D_t + D_x^2){F}^+\cdot G^-}{{F}^+G^-} = 0.
\end{align}
\end{subequations}
Multiplying \eqref{2ILW3a} and \eqref{2ILW3b} by $F^-{G}^+$ and ${F}^+G^-$, respectively, we conclude that \eqref{2ilw} is equivalent to the bilinear system \eqref{hirota_form}.
This completes the proof. 
\end{proof}

\subsection{Proof of Theorem \ref{bilinearth}B}
We decompose the solution $u,v$ of \eqref{2ilw} as $u=u_++u_-$, $v=v_++v_-$, with $u_{\pm}$, $v_{\pm}$ as in 
\eqref{uvprojections}. We view \eqref{FG_to_uv} as a pair of differential-difference equations for $F$, $G$ and seek solutions satisfying
\begin{equation}\label{implicit_ansatz}
\begin{split}
&\ii \partial_z \log F(x-\ii\delta/2,t)=u_+(x,t),\qquad \ii\partial_z \log F(x+\ii\delta/2,t)=-v_+(x,t), \\
&\ii \partial_z \log G(x-\ii\delta/2,t)=v_-(x,t),\qquad \ii\partial_z \log G(x+\ii\delta/2,t)=-u_-(x,t),
\end{split}
\end{equation}
which implies
\begin{equation}\label{RHsystem}
\begin{split}
\begin{cases}\ii \partial_z \log F(x-\ii\delta/2,t)- \ii\partial_z\log F(x-\ii\delta/2,t)=u_{+}(x,t)+v_{+}(x,t), \\
\ii \partial_z  \log G(x-\ii\delta/2,t)- \ii\partial_z\log G(x+\ii\delta/2,t)=v_{-}(x,t)+u_{-}(x,t).
\end{cases}
\end{split}
\end{equation}

We note that, using the hyperbolic identity for $\csch\, z\coloneqq1/\sinh z$
\begin{equation*}
\csch \,2z=\frac12(\coth z-\tanh z),
\end{equation*}
the functions $u_{\pm}+v_{\pm}$ may be written as 
\begin{align*}
u_{\pm}+v_{\pm}=\frac12(u+v)\mp\frac{\ii}{2} \pvint_{\R} \csch\bigg(\frac{\pi(x'-x)}{\delta}\bigg)(u(x')-v(x'))\,\mathrm{d}x'.
\end{align*}
As $x\to\pm\infty$, we have $\csch\, x\to \pm 2e^{-|x|}+\mathrm{O}(e^{-3|x|})$ uniformly, so the second term decays rapidly. Hence the existence of $T(u_{\pm}+v_{\pm})$ and $\tilde{T}(u_{\pm}+v_{\pm})$ follows from the existence of $Tu$, $Tv$, $\tilde{T}u$, $\tilde{T}v$.

Let $\tilde{C}$ denote the cylinder $\C/\ii\delta \Z$, $\tilde{\pi}$ the natural projection $\C\to\tilde{C}$, and $\tilde{C}_0$ the image of $\im \,z=0$ under $\tilde{\pi}$. Then \eqref{RHsystem} defines a pair of RH problems for the functions $\partial_z\log F$ and $\partial_z\log G$ on $\tilde{C}$. The following lemma can be proved similarly to Lemma \ref{RHlemma}.

\begin{lemma}[RH problem on $\tilde{C}$ with a jump across $\tilde{C}_0$]\label{RHlemma2}
Let $J: \tilde{C}_0 \to \C$ be a continuous function such that
$$\int_\R J(x)\, \mathrm{d}x= \bar{J}.$$ 
Then the scalar RH problem
\begin{itemize}
\item $A:\tilde{C} \setminus C_0$ is analytic.
\item Across $\tilde{C}_0$, $A$ satisfies the jump condition
$$A^+(x) - A^-(x) = J(x), \qquad x \in \tilde{C}_0\cong \R$$
\item $A(z) = \mp \bar{J}+\mathrm{O}(z^{-1}) $ as $x\to \pm\infty$, $z=x+\ii y \in \tilde{C}$
\end{itemize}
has the unique solution
\begin{align}\label{rh_soln2}
A(x) = \frac{1}{2\delta \ii} \int_{\R} \coth\bigg(\frac{\pi(x'-z)}{\delta}\bigg)J(x') \, \mathrm{d}x', \qquad x \in \tilde{C} \setminus \tilde{C}_0.
\end{align}
Moreover, this solution satisfies
$$A^\pm(x) = 
\frac{(T_{\frac{\delta}{2}}J)(x)}{2\ii} \pm \frac{1}{2} J(x), \qquad  x \in \tilde{C}_0 \cong \R.
$$
where
\begin{equation}
(T_{\frac{\delta}{2}}f)(x)\coloneqq\frac1{\delta}\pvint_\R \coth\bigg(\frac{\pi(x'-x)}{\delta}\bigg)f(x')\,\mathrm{d}x'.
\end{equation}
\end{lemma}

We are now in a position to prove Theorem \ref{bilinearth}B.

\begin{proof}[Proof of Theorem \ref{bilinearth}B]

Lemma \ref{RHlemma2} shows that the unique solution to the RH problem: 
\begin{itemize}
\item $A:\tilde{C}\setminus\tilde{C}_0$ is analytic, 
\item across $\tilde{C}_0$, $A$ satisfies the jump condition,  
$$A^+(x)-A^-(x)=u_{\pm}(x)+v_{\pm}(x)$$
\item $A(z)=\mp \bar{J}+\mathrm{O}(z^{-1})$ as $x\to\pm\infty$, $z=x+\ii y\in \tilde{C}$, 
\end{itemize}
is given by
\begin{equation}\label{RH2solnA}
A(z)=\frac{1}{2\ii\delta}\int_\R \coth\bigg(\frac{\pi(x'-z)}{\delta}\bigg) (u_{\pm}(x')+v_{\pm}(x'))\,\mathrm{d}x'
\end{equation}
with corresponding boundary values
\begin{equation}\label{RH2solnAbv}
A^{\pm}(x)=\frac{T_{\frac{\delta}{2}}(u_{\pm}+v_{\pm})(x)}{2\ii}\pm\frac12(u_{\pm}(x)+v_{\pm}(x)).
\end{equation}
Hence, using $\coth(z+\ii\pi/2)=\tanh z$ in \eqref{RH2solnA}, we see that \eqref{FG_from_uv} verifies \eqref{RHsystem}. 

The boundary values of \eqref{FG_from_uv} follow from \eqref{RH2solnAbv}:
\begin{equation}
\begin{split}\label{logFGbv}
\begin{dcases}\ii\partial_z \log F(x\pm\ii\delta/2)=\frac{T_{\frac{\delta}{2}}[u_++v_+](x)}{2\ii}\pm \frac{1}{2}(u_+(x)+v_+(x)), \\
\ii\partial_z \log G(x\pm\ii\delta/2)=\frac{T_{\frac{\delta}{2}}[u_-+v_-](x)}{2\ii}\pm \frac{1}{2}(u_-(x)+v_-(x)). 
\end{dcases}
\end{split}
\end{equation}
We next show that \eqref{implicit_ansatz} is satisfied. Using the hyperbolic identity
\begin{equation*}
\coth 2z=\frac12(\coth z+\tanh z),
\end{equation*}
we see that $T_\frac{\delta}{2}$ can be written as $T_{\frac{\delta}{2}}=T+\tilde{T}$. Hence,
\begin{align*}
T_{\frac{\delta}{2}}(u_{\pm}+v_{\pm})=&(T+\tilde{T})(u_{\pm}+v_{\pm})\\
=&\frac12(T+\tilde{T})(u+v)\mp\frac{\ii}{2}(T+\tilde{T})(T-\tilde{T})(u-v).
\end{align*}
Then the identity $(T+\tilde{T})(T-\tilde{T})f=-f$ from Proposition \ref{TpropertiesR}, shows that
\begin{align*}
T_{\frac{\delta}{2}}(u_{\pm}+v_{\pm})=\frac12(T+\tilde{T})(u+v)\pm\frac{\ii}{2}(u-v)=\pm\ii(u_{\pm}-v_{\pm});
\end{align*}
it follows from \eqref{logFGbv} that \eqref{implicit_ansatz} is satisfied. Thus, \eqref{FG_to_uv} holds and Theorem \ref{bilinearth}A shows that \eqref{hirota_form} is satisfied. 
\end{proof}

\begin{remark}
We note that the ansatz
\begin{align}
F(x,t)=\prod_{j=1}^N \sinh \bigg(\frac{\pi}{2\delta}(z-z_j(t))\bigg),\quad G(x,t)=\prod_{j=1}^M \sinh \bigg(\frac{\pi}{2\delta}(z-w_j(t))\bigg),
\end{align}
for \eqref{hirota_form}, together with Theorem \ref{bilinearth}A, leads to an alternative proof of the following result in \cite{berntson2020a}. One can get real-valued solutions $u(x,t)$ and $v(x,t)$ by restricting to $M=N$ and $b_j=\bar a_j$ for $j=1,\ldots,N$ below.

\begin{corollary}[Soliton solutions of the non-chiral ILW equation]\label{solitoncorollary}
For arbitrary non-negative integers $N,M$ and complex parameters $a_j,$ $j=1,\ldots,N$ and $b_j,$ $j=1,\ldots,M,$ satisfying
\begin{equation*}
\im(a_j\pm\ii\delta/2)\neq 2\delta n,\quad \im(b_j\pm\ii\delta/2)\neq 2\delta n,
\end{equation*}
for all integers $n$, the functions
\begin{align*}
\begin{dcases}
u(x,t)=\frac{\ii\pi}{2\delta} \sum_{j=1}^N \coth\bigg(\frac{\pi}{2\delta}(z-z_j(t)-\ii\delta/2)  \bigg)- \frac{\ii\pi}{2\delta} \sum_{j=1}^M \coth\bigg(\frac{\pi}{2\delta}(z-w_j(t)+\ii\delta/2)  \bigg)    \\
v(x,t)=-\frac{\ii\pi}{2\delta} \sum_{j=1}^N \coth\bigg(\frac{\pi}{2\delta}(z-z_j(t)+\ii\delta/2)  \bigg)+\frac{\ii\pi}{2\delta} \sum_{j=1}^M \coth\bigg(\frac{\pi}{2\delta}(z-w_j(t)-\ii\delta/2)  \bigg)   
\end{dcases}
\end{align*}
provide a solution of the non-chiral ILW equation \eqref{2ilw} provided the poles $z_j(t)$ and $w_j(t)$ satisfy
\begin{align*}
\begin{dcases}\ddot{z}_j=-\frac{\pi^2}{\delta^2}\underset{k\neq j}{\sum_{j=1}^N} \csch^2\bigg(\frac{\pi}{\delta}(z_j-z_k)\bigg),\qquad \im(z_j\pm\ii\delta/2)\neq 2\delta n, \\
\ddot{w}_j=-\frac{\pi^2}{\delta^2}\underset{k\neq j}{\sum_{j=1}^M} \csch^2\bigg(\frac{\pi}{2\delta}(w_j-w_k)\bigg),\qquad \im(w_j\pm\ii\delta/2)\neq 2\delta n,
\end{dcases}
\end{align*}
with initial conditions
\begin{align*}
&z_j(0)=a_j,\quad w_j(0)=b_j, \\
&\begin{dcases}\dot{z}_j(0)=\frac{\ii\pi}{\delta}\underset{k\neq j}{\sum_{k=1}^N} \coth\bigg(\frac{\pi}{2\delta}(a_j-a_k)\bigg)-\frac{\ii\pi}{\delta}\sum_{k=1}^M \coth\bigg(\frac{\pi}{2\delta}(a_j-b_k+\ii\delta)\bigg) \\
\dot{w}_j(0)=-\frac{\ii\pi}{\delta}\underset{k\neq j}{\sum_{k=1}^M} \coth\bigg(\frac{\pi}{2\delta}(b_j-b_k)\bigg)+\frac{\ii\pi}{\delta}\sum_{k=1}^M \coth\bigg(\frac{\pi}{2\delta}(b_j-a_k+\ii\delta)\bigg).
\end{dcases}
\end{align*}
\end{corollary}
\end{remark}

\section{B\"{a}cklund transformation}\label{backlundsec}
Our derivation of a B\"{a}cklund transformation for \eqref{2ilw} is inspired by the analogous derivation for the Benjamin-Ono equation presented in \cite[Chapter 3]{matsuno1984}.

Suppose $(u,v)$ and $(\tilde{u},\tilde{v})$ are two solutions of \eqref{2ilw} with associated Hirota bilinear forms \eqref{hirota_form} and
\begin{subequations}\label{hirota_F2G2}
\begin{align}
&(\ii D_t-D_x^2)\tilde{F}^-\cdot \tilde{G}^+=0,\label{hirota_F2G2_a} \\
&(\ii D_t-D_x^2)\tilde{F}^+\cdot \tilde{G}^-=0,\label{hirota_F2G2_b}
\end{align}
\end{subequations}
respectively, where
\begin{align*}
&u=\ii\partial_x\log \frac{F^+}{G^-},\qquad v=\ii\partial_x\log \frac{G^-}{F^+}, \\
&\tilde{u}=\ii\partial_x\log \frac{\tilde{F}^+}{\tilde{G}^-},\qquad \tilde{v}=\ii\partial_x\log \frac{\tilde{G}^-}{\tilde{F}^+}.
\end{align*}
Then, in terms of the variables $F$, $G$, $\tilde{F}$, $\tilde{G}$, the B\"{a}cklund transformation of \eqref{2ilw} is given by 
\begin{subequations}
\label{backlund_hirota}
\begin{align}
&(\ii D_t-2\ii\alpha D_x-D_x^2-\gamma)F^-\cdot \tilde{F}^-=0, \label{backlund_hirota_a}\\
&(\ii D_t-2\ii\alpha D_x-D_x^2-\gamma)G^+\cdot \tilde{G}^+=0, \label{backlund_hirota_b}\\
&(D_x+\ii\alpha)G^+\cdot \tilde{F}^-=\ii\beta F^-\cdot \tilde{G}^+,  \label{backlund_hirota_c}\\
&(\ii D_t-2\ii\alpha D_x-D_x^2-\gamma)F^+\cdot\tilde{F}^+=0,      \label{backlund_hirota_d} \\
&(\ii D_t-2\ii\alpha D_x-D_x^2-\gamma)G^-\cdot\tilde{G}^-=0, \label{backlund_hirota_e}\\
&(D_x+\ii\alpha)G^-\cdot\tilde{F}^+=\ii\beta F^+\cdot \tilde{G}^-,\label{backlund_hirota_f}
\end{align}
\end{subequations}
where $\alpha,\beta,\gamma\in\C$ are arbitrary constants. 

\begin{proposition}[B\"{a}cklund transformation in terms of bilinear variables]\label{backlundbilinear}
Suppose $(F,G)$ and $(\tilde{F},\tilde{G})$ satisfy the relations in \eqref{backlund_hirota}. Then $(F,G)$ is a solution of \eqref{hirota_form} if and only $(\tilde{F},\tilde{G})$ is a solution of \eqref{hirota_F2G2}. 
\end{proposition}
\begin{proof}
Suppose the relations in \eqref{backlund_hirota} hold and that $(F,G)$ satisfy \eqref{hirota_form}. We will show that $(\tilde{F},\tilde{G})$ satisfies \eqref{hirota_F2G2}. The reverse implication then follows by symmetry. 

Consider the quantity
\begin{align}
Q=&\big((\ii D_t-D_x^2)F^-\cdot G^+\big)\tilde{F}^-\tilde{G}^+-F^-G^+\big((\ii D_t-D_x^2)\tilde{F}^-\cdot\tilde{G}^+\big).  \label{Q_def}
\end{align}
To simplify this, we need a pair of identities \cite[Appendix I]{matsuno1984}:
\begin{subequations}\label{hirota_identities}
\begin{align}\label{hirota_identity1}
&(D_x a\cdot b)cd-ab(D_x c\cdot d)=(D_x a\cdot c)bd-ac(D_x b\cdot d), \\
\label{hirota_identity2}
&(D_x^2a\cdot b)cd-ab(D_x^2c\cdot d)=D_x\big((D_t a\cdot d)\cdot(bc)+(ad)\cdot(D_tb\cdot d)\big).
\end{align}
\end{subequations}
Using \eqref{hirota_identity1}, we have 
\begin{align*}
Q=&(\mathrm{i}D_tF^-\cdot\tilde{F}^-)G^+\tilde{G}^+-F^-\tilde{F}^-(\mathrm{i}D_tG^-\cdot\tilde{G}^-)\\
&-(D_x^2F^-\cdot G^+)\tilde{F}^-\tilde{G}^++F^-G^+(D_x^2\tilde{F}^-\cdot\tilde{G}^+). 
\end{align*}
We now use \eqref{backlund_hirota_a} and \eqref{backlund_hirota_b} to write
\begin{align*}
Q=&2\ii\alpha\big( (D_x F^-\cdot \tilde{F}^-)G^+\tilde{G}^+- F^-\tilde{F}^- (D_xG^+\cdot\tilde{G}^+)     \big) \\
&+\big( (D_x^2 F^-\cdot \tilde{F}^-)G^+\tilde{G}^+- F^-\tilde{F}^- (D_x^2G^+\cdot\tilde{G}^+)      \big) \\
&-\big( (D_x^2 F^-\cdot G^+)\tilde{F}^-\tilde{G}^+- F^-G^+ (D_x^2\tilde{F}^-\cdot\tilde{G}^+)      \big).
\end{align*}
Using both identities in \eqref{hirota_identities}, we find
\begin{align*}
Q=&2\ii\alpha D_x (F^-\tilde{G}^+)\cdot (\tilde{F}^-G^+)\\
&+D_x\big((D_x F^-\cdot \tilde{G}^+)\cdot(G^+\tilde{F}^-)+(F^-\tilde{G}^+)\cdot (D_x  G^+\cdot \tilde{F}^-)  \big) \\
&-D_x\big((D_x F^-\cdot \tilde{G}^+)\cdot(\tilde{F}^-G^+)+(F^-\tilde{G}^+)\cdot (D_x \tilde{F}^-\cdot G^+) \big) \\
=&2D_x\big( (F^-\tilde{G}^+)\cdot(\ii\alpha+D_x)G^+\cdot\tilde{F}^-           \big).
\end{align*}
Finally, using \eqref{backlund_hirota_a}, we see that $Q=0$. Thus, recalling \eqref{Q_def}, we see that \eqref{hirota_form_a} is satisfied if and only if \eqref{hirota_F2G2_a} is satisfied. The proof for \eqref{hirota_form_b} and \eqref{hirota_F2G2_b} is similar. 
\end{proof}

\subsection{B\"{a}cklund transformation in terms of $u$, $v$, $\tilde{u}$, $\tilde{v}$}

To transform \eqref{backlund_hirota} into a form written in the original variables, we introduce potential functions $U$, $V$, $\tilde{U}$, $\tilde{V}$ by
\begin{subequations}\label{UVdef}
\begin{align}
U\coloneqq\ii\log \frac{F^-}{G^+},\qquad V\coloneqq\ii\log\frac{G^-}{F^+}, \\
\tilde{U}\coloneqq\ii\log \frac{\tilde{F}^-}{\tilde{G}^+},\qquad \tilde{V}\coloneqq\ii\log\frac{\tilde{G}^-}{\tilde{F}^+} ,
\end{align}
\end{subequations}
so that
\begin{equation*}
U_x=u, \qquad V_x=v,\qquad \tilde{U}_x=\tilde{u},\qquad \tilde{V}_x=\tilde{v}. 
\end{equation*}

\begin{theorem}[B\"{a}cklund transformation for the non-chiral ILW equation]\label{backlundth}
Suppose the following relations hold:
  \begin{subequations}\label{uvbacklund}
\begin{align}\label{uvbacklunda}
& u = \frac{1 - e^{-W}}{\epsilon} - \ii P_- W_x - \frac{1}{2}\tilde{T}Z_x,
	\\ \label{uvbacklundb}
& W_t = -\frac{2}{\epsilon}(1 - e^{-W})W_x - TW_{xx} - \tilde{T}Z_{xx}
 + W_xTW_x + W_x \tilde{T}Z_x,
	\\ \label{uvbacklundc}
& v = -\frac{1- e^Z}{\epsilon} + \ii P_+Z_x + \frac{1}{2}\tilde{T}W_x,
	\\ \label{uvbacklundd}
& Z_t = -\frac{2}{\epsilon} (1- e^{Z}) Z_x + TZ_{xx} + \tilde{T}W_{xx}
 + Z_xTZ_x + Z_x\tilde{T}W_x,
\end{align}
\end{subequations}
where
\begin{equation}\label{WZ}
W\coloneqq \ii(U - \tilde{U}),\qquad  Z\coloneqq\ii(V - \tilde{V}),
\end{equation}
and 
\begin{equation}\label{Ppm}
P_{\pm}\coloneqq-\frac12(\ii T\pm 1).
\end{equation} 
Then $(u,v)$ satisfy the non-chiral ILW equation \eqref{2ilw} if and only if $(\tilde{u}, \tilde{v})$ do.
\end{theorem}

\subsection{Proof of Theorem \ref{backlundth}}
We show that the equations \eqref{backlund_hirota} take the form \eqref{uvbacklund} when rewritten in terms of $u$, $v$, $\tilde{u}$, $\tilde{v}$.
Let us first rewrite \eqref{backlund_hirota_c}. Dividing \eqref{backlund_hirota_c} by $G^+\tilde{F}^-$ yields
$$\frac{{G}^+_{x}}{G^+} - \frac{\tilde{F}^-_{x}}{\tilde{F}^-} + \ii\alpha = \ii\beta \frac{F^- \tilde{G}^+}{\tilde{F}^- G^+},$$
i.e.,
\begin{align}\label{fgbacklundcrewrite}
u_- + \tilde{u}_+ = -\alpha + \beta e^{-\ii(U - \tilde{U})},
\end{align}
with $u_{\pm},$ $v_{\pm}$ as in \eqref{uvplusminusdef} and where $\tilde{u}_{\pm},$ $\tilde{v}_{\pm}$ are defined analogously. 

\begin{lemma}\label{uPlemma}
The following identities hold:
  \begin{align}
\begin{cases}
    u_+ = P_- u - \frac{\ii}{2}\tilde{T}v, \\
    u_- = -P_- u + \frac{\ii}{2}\tilde{T}v + u,
    \end{cases} \qquad
  \begin{cases}
    \tilde{u}_+ = P_- \tilde{u} - \frac{\ii}{2}\tilde{T}\tilde{v}, \\
    \tilde{u}_- = -P_- \tilde{u} + \frac{\ii}{2}\tilde{T}\tilde{v} + \tilde{u}.
    \end{cases}   
  \end{align}
\end{lemma}
\begin{proof}
By Lemma \ref{lemma2},
$$Tu + \tilde{T}v + \ii u = 2\ii u_+$$
and the expression for $u_+$ follows after simplification. The expression for $u_-$ then follows because $u = u_+ + u_-$.
The expressions for $\tilde{u}_\pm$ follow in the same way. 
\end{proof}

Utilizing Lemma \ref{uPlemma}, equation \eqref{fgbacklundcrewrite} can be rewritten as
\begin{align}\label{Pminusutildeu}
-P_- (u-\tilde{u}) + \frac{\ii}{2}\tilde{T}(v - \tilde{v}) + u = -\alpha+ \beta e^{-\ii(U - \tilde{U})}.
\end{align}
Using \eqref{WZ} and setting
and
\begin{equation}\label{alphabetavalue}
\alpha=-\frac{1}{\epsilon}, \quad \beta =-\frac{1}{\epsilon},
\end{equation}
this yields
\begin{align*}
u = \frac{1 - e^{-W}}{\epsilon} - \ii P_- W_x - \frac{1}{2}\tilde{T}Z_x,
\end{align*}
which is \eqref{uvbacklunda}.

We next rewrite the $t$-parts \eqref{backlund_hirota_a}--\eqref{backlund_hirota_b} of the B\"acklund transformation as
\begin{align*}
\big(\ii \partial_t - 2\ii\alpha\partial_x)\log\frac{F^-}{\tilde{F}^-} - \partial_x^2 \log(F^-\tilde{F}^-) - \bigg(\partial_x\log\frac{F^-}{\tilde{F}^-}\bigg)^2 - \gamma = 0,\\
\big(\ii \partial_t - 2\ii\alpha\partial_x)\log\frac{G^+}{\tilde{G}^+} - \partial_x^2 \log(G^+\tilde{G}^+) - \bigg(\partial_x\log\frac{G^+}{\tilde{G}^+}\bigg)^2 - \gamma= 0.
\end{align*}
Subtracting the second of these equations from the first gives
\begin{equation}
\begin{split}\label{ipartalt2ilambda}
& \big(\ii \partial_t - 2\ii\alpha\partial_x)\bigg(\log\frac{F^-}{G^+} - \log \frac{\tilde{F}^-}{\tilde{G}^+}\bigg)
- \bigg(\log\frac{F^-}{G^+} + \log\frac{\tilde{F}^-}{\tilde{G}^+}\bigg)_{xx}
	\\ 
& - \bigg(\log\frac{F^-}{G^+}- \log\frac{\tilde{F}^-}{\tilde{G}^+}\bigg)_x\big(\log({F^-}{G^+}) - \log({\tilde{F}^-}{\tilde{G}^+})\big)_x = 0. 
\end{split}
\end{equation}
Multiplying by $\ii$ and using the definitions \eqref{uvplusminusdef} and \eqref{UVdef} of $u_\pm$, $\tilde{u}_\pm$ and $U, \tilde{U}$, this becomes
\begin{align*}
& \big(\ii \partial_t - 2\ii\alpha\partial_x)(U - \tilde{U}) - (U + \tilde{U})_{xx}
 +\ii (U - \tilde{U})_x\big(u_+ - u_- - (\tilde{u}_+ - \tilde{u}_-)\big) = 0. 
\end{align*}
Recalling \eqref{WZ} and using Lemma \ref{lemma2}, we find
\begin{align*}
& W_t - 2\alpha W_x - (U + \tilde{U})_{xx}
 - \ii W_x\big(TU + \tilde{T}V - T\tilde{U} - \tilde{T}\tilde{V}\big)_x = 0. 
\end{align*}
Equation \eqref{Pminusutildeu} can be written as
$$\frac{1}{2}(U + \tilde{U})_x = -\alpha + \beta e^{-W} - \frac{1}{2} TW_x - \frac{1}{2} \tilde{T}Z_x.$$
Using this relation to eliminate $(U + \tilde{U})_{xx}$, we arrive at 
\begin{align*}
& W_t - 2\alpha W_x + 2\beta W_x e^{-W} + TW_{xx} + \tilde{T}Z_{xx}
 - W_x\big(TW_x + \tilde{T}Z_x\big) = 0. 
\end{align*}
That is,
\begin{align*}
& W_t = -\frac{2}{\epsilon}(1 - e^{-W})W_x - TW_{xx} - \tilde{T}Z_{xx}
 + W_xTW_x + W_x \tilde{T}Z_x,
\end{align*}
which is \eqref{uvbacklundb}.

We next rewrite the $x$-part \eqref{backlund_hirota_f}. Dividing \eqref{backlund_hirota_f} by $G^-\tilde{F}^+$ yields
$$\frac{{G}^-_{x}}{G^-} - \frac{{F}^+_{2,x}}{\tilde{F}^+} + \ii\alpha = \ii\beta \frac{F^+ \tilde{G}^-}{\tilde{F}^+ G^-},$$
i.e.,
\begin{align}\label{fgbacklundfrewrite}
v_-  + \tilde{v}_+ = \alpha - \beta e^{\ii(V - \tilde{V})}.
\end{align}

\begin{lemma}\label{vPlemma}
The following identities hold:
  \begin{align}
\begin{cases}
    v_+ = -P_+ v + \frac{\ii}{2}\tilde{T}u, \\
    v_- = P_+ v - \frac{\ii}{2}\tilde{T}u + v,
    \end{cases} \qquad
  \begin{cases}
      \tilde{v}_+ = -P_+ \tilde{v} + \frac{\ii}{2}\tilde{T}\tilde{u}, \\
    \tilde{v}_- = P_+ \tilde{v} - \frac{\ii}{2}\tilde{T}\tilde{u} + \tilde{v}.
  \end{cases}   
  \end{align}
\end{lemma}
\begin{proof}
By Lemma \ref{lemma2},
$$Tv + \tilde{T}u - \ii v= -2\ii v_+$$
and the expression for $v_+$ follows after simplification. The expression for $v_-$ then follows because $v = v_+ + v_-$.
The expressions for $\tilde{v}_\pm$ follow in the same way. 
\end{proof}

Utilizing Lemma \ref{vPlemma}, equation \eqref{fgbacklundfrewrite} can be rewritten as
\begin{align}\label{Pminusvtildev}
P_+ (v- \tilde{v}) - \frac{\ii}{2}\tilde{T}(u - \tilde{u}) + v
= \alpha - \beta e^{\ii(V - \tilde{V})}.
\end{align}
With \eqref{WZ} and \eqref{alphabetavalue}
this becomes
\begin{align*}
v = -\frac{1- e^Z}{\epsilon} + \ii P_+Z_x + \frac{1}{2}\tilde{T}W_x,
\end{align*}
which is \eqref{uvbacklundc}.

We next rewrite the $t$-parts \eqref{backlund_hirota_d}--\eqref{backlund_hirota_e} of the B\"acklund transformation.
As before, we find that \eqref{ipartalt2ilambda} holds except that $F,\tilde{F}$ and $G,\tilde{G}$ are now evaluated at $x + \ii\delta/2$ and $x - \ii\delta/2$, respectively, i.e.,
\begin{equation}
\begin{split}
& \big(\ii \partial_t - 2\ii\alpha\partial_x)\bigg(\log\frac{{F^+}}{G^-} - \log \frac{\tilde{F}^+}{\tilde{G}^-}\bigg)
- \bigg(\log\frac{F^+}{G^-} + \log\frac{\tilde{F}^+}{\tilde{G}^-}\bigg)_{xx}
	\\
& - \bigg(\log\frac{F^+}{G^-}- \log\frac{\tilde{F}^+}{\tilde{G}^-}\bigg)_x\big(\log({F^+}{G^-}) - \log({\tilde{F}^+}{\tilde{G}^-})\big)_x = 0. 
\end{split}
\end{equation}
Multiplying by $\ii$ and using the definitions \eqref{uvplusminusdef} and \eqref{UVdef} of $v^\pm$, $\tilde{v}^\pm$ and $V, \tilde{V}$, this becomes
\begin{align}\nonumber
& \big(\ii \partial_t - 2\ii\alpha\partial_x)(-V + \tilde{V}) 
+ (V + \tilde{V})_{xx}
 + \ii (-V + \tilde{V})_x\big(-v_+ + v_- + \tilde{v}_+ - \tilde{v}_-\big) = 0. 
\end{align}
Recalling \eqref{WZ} and using Lemma \ref{lemma2}, we find
\begin{align*}
& -Z_t + 2\alpha Z_x + (V + \tilde{V})_{xx}
 + \ii Z_x\big(TV + \tilde{T}U - T\tilde{V} - \tilde{T}\tilde{U}\big)_x = 0. 
\end{align*}
Equation \eqref{Pminusvtildev} can be written as
$$\frac{1}{2} (V + \tilde{V})_x  
= \alpha - \beta e^{Z} + \frac{1}{2} TZ_x + \frac{1}{2}\tilde{T}W_x.$$
Using this relation to eliminate $(V + \tilde{V})_{xx}$, we arrive at 
\begin{align*}
& -Z_t + 2\alpha Z_x - 2\beta Z_x e^{Z} + TZ_{xx} + \tilde{T}W_{xx}
 + Z_x\big(TZ_x + \tilde{T}W_x\big) = 0. 
\end{align*}
That is,
\begin{align*}
Z_t = -\frac{2}{\epsilon} (1- e^{Z}) Z_x + TZ_{xx} + \tilde{T}W_{xx}
 + Z_xTZ_x + Z_x\tilde{T}W_x,
\end{align*}
which is \eqref{uvbacklundd}.
This completes the  proof of Theorem \ref{backlundth}.

\section{Conservation laws}\label{conservationsec}

In this section, we provide two complementary proofs of the following theorem.
\begin{theorem}[Conservation laws of the non-chiral ILW equation]\label{conservation_laws}
The non-chiral ILW equation \eqref{2ilw} has an infinite number of conservation laws
\begin{equation}
I_n=\int_\R (W_n+Z_n)\,\mathrm{d}x,
\end{equation}
where $W_n$ and $Z_n$ can be computed recursively from the following formal power series in $\epsilon$, 
\begin{subequations}\label{WZimplicit}
\begin{align}
& u = \frac{1 - \exp\bigg(-\sum\limits_{n=1}^\infty W_n\epsilon^n\bigg) }{\epsilon} -\ii P_- \sum_{n=1}^\infty W_{n,x} \epsilon^n - \frac{1}{2}\tilde{T}\sum_{n=1}^\infty Z_{n,x} \epsilon^n,
	\\
& v = -\frac{1- \exp\bigg(\sum\limits_{n=1}^\infty Z_{n} \epsilon^n\bigg)}{\epsilon} +\ii P_+\sum_{n=1}^\infty Z_{n,x} \epsilon^n  + \frac{1}{2}\tilde{T}\sum_{n=1}^\infty W_{n,x} \epsilon^n,
\end{align}
\end{subequations}
with $P_{\pm}$ as in \eqref{Ppm}. The first four conservation laws are 
\begin{subequations}\label{3conservation_laws}
\begin{align}
&I_{1,u}=\int_\R u\,\mathrm{d}x, \quad  I_{1,v}=\int_\R v\,\mathrm{d}x,  \label{conservation_law1}\\
&I_2= \frac12\int_\R (u^2-v^2)\,\mathrm{d}x, \label{conservation_law2}\\
&I_3=\int_\R \bigg( \frac13(u^3+v^3)+\frac12( uTu_x+vTv_x+u\tilde{T}v_x+v\tilde{T}u_x    \bigg)\mathrm{d}x. \label{conservation_law3} \\
&I_4=\int_\R \bigg(\frac{u^4-v^4}{4}+\frac{u_x^2-v_x^2}{8}+\frac38\big( (Tu_x)^2-(Tv_x)^2-(\tilde{T}u_x)^2+(\tilde{T}v_x)^2\big) \label{conservation_law4}\\
&+\frac34\big(u^2 Tu_x-v^2 Tv_x\big)+\frac34\big(u^2\tilde{T}v_x-v^2\tilde{T}u_x\big)\bigg)\,\mathrm{d}x.\nonumber
\end{align}
\end{subequations}
\end{theorem}
The first proof uses the B\"{a}cklund transformation \eqref{uvbacklund} while the second proof uses the Lax pair \eqref{2ilw_laxpair} and illustrates its utility. In both proofs we construct a conservation law with dependence on an auxiliary parameter; an infinite number of conservation laws are obtained by an appropriate expansion in this parameter.
At the end of this section, we verify by direct computation that the three first three quantities in \eqref{3conservation_laws} are conserved. 

\subsection{Proof of Theorem \ref{conservation_laws} using the B\"{a}cklund transformation}\label{conservationslawsBacklundproof}

Adding equations \eqref{uvbacklundb} and \eqref{uvbacklundd}, we find
\begin{align*}
W_t + Z_t = & -\frac{2}{\epsilon}(1 - e^{-W})W_x - TW_{xx} - \tilde{T}Z_{xx}
 + W_xTW_x + W_x \tilde{T}Z_x
 	\\
& -\frac{2}{\epsilon} (1- e^{Z}) Z_x + TZ_{xx} + \tilde{T}W_{xx}
 + Z_xTZ_x + Z_x\tilde{T}W_x.
\end{align*}
Thus,
$$\frac{\mathrm{d}}{\mathrm{d}t} \int_\R (W + Z)\,\mathrm{d}x
= \int_\R \big(W_xTW_x + W_x \tilde{T}Z_x + Z_xTZ_x + Z_x\tilde{T}W_x\big)\,\mathrm{d}x.$$
Using the anti-self-adjointness \eqref{anti_self_adjointTTtilde} of the operators $T$ and $\tilde{T}$, the right-hand side vanishes. 
Hence $\int_\R (W + Z)\,\mathrm{d}x$ is a conserved quantity. 
If we expand $W$ and $Z$ formally in powers of $\epsilon$ as
\begin{align}\label{WZexpansions}
W = \sum_{n=1}^\infty W_n \epsilon^n,  \qquad Z = \sum_{n=1}^\infty Z_n \epsilon^n,
\end{align}
we find that
\begin{align}\label{sumdIndt}
\sum_{n=1}^\infty \frac{\mathrm{d}I_n}{\mathrm{d}t} \epsilon^n = 0,
\end{align}
where
\begin{align}\label{Indef}
I_n = \int_\R (W_n + Z_n)\, \mathrm{d}x.
\end{align}
The identity \eqref{sumdIndt} must hold for arbitrary $\epsilon$, therefore
$$\frac{\mathrm{d}I_n}{\mathrm{d}t} = 0, \qquad n \in \mathbb{N}.$$

The $I_n$ defined in \eqref{Indef} is the $n$th conserved quantity of \eqref{2ilw}. To derive the explicit functional forms of $W_n$ and $Z_n$, we substitute \eqref{WZexpansions} into \eqref{uvbacklunda} and \eqref{uvbacklundc}, which gives \eqref{WZimplicit}
Note that
\begin{align*}
\frac{1 - \exp\bigg(-\sum\limits_{n=1}^\infty W_n\epsilon^n\bigg) }{\epsilon}
= &\; W_1 + \frac{-W_1^2+ 2W_2}{2} \epsilon + \frac{W_1^3 - 6W_1 W_2 + 6W_3}{6}\epsilon^2 
	\\
& + \frac{-W_1^4 + 12W_1^2W_2 - 12 W_2^2 - 24W_1 W_3 + 24 W_4}{24} \epsilon^3 + \mathrm{O}(\epsilon^4),
\end{align*}
and
\begin{align*}
-\frac{1- \exp\bigg(\sum\limits_{n=1}^\infty Z_{n} \epsilon^n\bigg)}{\epsilon}
= &\; Z_1 + \frac{Z_1^2+ 2Z_2}{2} \epsilon + \frac{Z_1^3 + 6Z_1 Z_2 + 6Z_3}{6}\epsilon^2 
	\\
& + \frac{Z_1^4 + 12Z_1^2Z_2 + 12 Z_2^2 + 24 Z_1 Z_3 + 24 Z_4}{24} \epsilon^3 + \mathrm{O}(\epsilon^4).
\end{align*}
The terms of $\mathrm{O}(1)$ yield
\begin{align}\label{W1Z1}
& u = W_1, \qquad v = Z_1,
\end{align}
which gives the conservation law
\begin{equation}
I_1=\int_\R (u+v)\,\mathrm{d}x.
\end{equation}
In fact the quantities $I_{1,u}$ and $I_{2,v},$ defined in \eqref{conservation_law1}, are individually conserved. We verify this by direct computation in Section \ref{directverification} below.

The terms of $\mathrm{O}(\epsilon)$ give
\begin{align*}
& 0 = \frac{-W_1^2+ 2W_2}{2} - \ii P_- W_{1,x} - \frac{1}{2}\tilde{T}Z_{1,x},
	\\
& 0 = \frac{Z_1^2+ 2Z_2}{2} + \ii P_+Z_{1,x} + \frac{1}{2}\tilde{T} W_{1,x},
\end{align*}
that is,
\begin{align}
\begin{split}\label{W2Z2}
& W_2 = \frac{u^2}{2} + \ii P_- u_x + \frac{1}{2}\tilde{T}v_x,
\qquad Z_2 = -\frac{v^2}{2} - \ii P_+ v_x - \frac{1}{2}\tilde{T} u_x,
\end{split}
\end{align}
leading to the expression \eqref{conservation_law2} for $I_2$. 
The terms of $\mathrm{O}(\epsilon^2)$ give
\begin{align*}
& 0 = \frac{W_1^3 - 6W_1 W_2 + 6W_3}{6} - \ii P_- W_{2,x} - \frac{1}{2}\tilde{T}Z_{2,x},
	\\
& 0 = \frac{Z_1^3 + 6Z_1 Z_2 + 6Z_3}{6} + \ii P_+Z_{2,x} + \frac{1}{2}\tilde{T} W_{2,x},
\end{align*}
i.e., 
\begin{equation}
\begin{split}\label{W3Z3}
& W_3 = - \frac{u^3 - 6u W_2}{6} + \ii P_- W_{2,x} + \frac{1}{2}\tilde{T}Z_{2,x},
	\\
& Z_3 = -\frac{v^3 + 6v Z_2}{6} - \ii P_+Z_{2,x} - \frac{1}{2}\tilde{T} W_{2,x}.
\end{split}
\end{equation}
After some simplifications, this gives the expression \eqref{conservation_law3} for $I_3$:
\begin{align*}
I_3 = &\; \int_\R \bigg(- \frac{u^3 - 6u W_2}{6} - \frac{v^3 + 6v Z_2}{6}\bigg) \,\mathrm{d}x
	\\
= &\; \int_\R \bigg(- \frac{u^3 + v^3}{6} + u W_2 - v Z_2\bigg) \,\mathrm{d}x
	\\
= &\; \int_\R \bigg(- \frac{u^3 + v^3}{6} 
+ \frac{u^3}{2} + \ii uP_- u_x + \frac{u}{2}\tilde{T}v_x
+ \frac{v^3}{2} + \ii vP_+ v_x + \frac{v}{2}\tilde{T} u_x\bigg) \,\mathrm{d}x
	\\
= &\; \int \bigg(\frac{1}{3}(u^3 + v^3)  + \frac{1}{2}(u T u_x + v T v_x + u\tilde{T}v_x+ v\tilde{T} u_x)\bigg) \,\mathrm{d}x.
\end{align*}
The terms of $\mathrm{O}(\epsilon^3)$ give
\begin{align*}
& 0 = \frac{-W_1^4 + 12W_1^2W_2 - 12 W_2^2 - 24W_1 W_3 + 24 W_4}{24} - \ii P_- W_{3,x} - \frac{1}{2}\tilde{T}Z_{3,x},
	\\
& 0 = \frac{Z_1^4 + 12Z_1^2Z_2 + 12 Z_2^2 + 24 Z_1 Z_3 + 24 Z_4}{24} + \ii P_+Z_{3,x} + \frac{1}{2}\tilde{T} W_{3,x},
\end{align*}
i.e., 
\begin{equation}
\begin{split}\label{W4Z4}
& W_4 = -\frac{-W_1^4 + 12W_1^2W_2 - 12 W_2^2 - 24W_1 W_3}{24} + \ii P_- W_{3,x} + \frac{1}{2}\tilde{T}Z_{3,x},
	\\
& Z_4 = -\frac{Z_1^4 + 12Z_1^2Z_2 + 12 Z_2^2 + 24 Z_1 Z_3}{24} - \ii P_+Z_{3,x} - \frac{1}{2}\tilde{T} W_{3,x}.
\end{split}
\end{equation}
Then, a lengthy calculation using \eqref{W1Z1}--\eqref{W3Z3} and the identities \eqref{anti_self_adjointTTtilde}--\eqref{Tcommutator} in \eqref{Indef} gives \eqref{conservation_law4}.

\subsection{Proof of Theorem \ref{conservation_laws} using the Lax pair}

We generalize the approach of \cite{ablowitz1982}. It is first necessary to transform the Lax pair \eqref{2ilw_laxpair} into a more convenient form and define particular eigenfunctions. 

We define the functions
\begin{equation}\label{W_defn}
W_1(z;k)=\psi(z+\ii\delta/2)e^{\ii kz},\qquad W_2(z;k)=\psi(z+3\ii\delta/2)e^{\ii k z}.
\end{equation}
We view $W_1(z)$ and $W_2(z)$ as analytic functions on the strip $0<\im z<\delta$ and we use the notation $W^{\pm}$ for the continous boundary values of these functions as $z$ approaches the lines $C_0$ (from above) and $C_\delta$ (from below), respectively. In terms of \eqref{W_defn}, the $x$-part of the Lax pair \eqref{2ilw_laxpair} is written as
\begin{align}\label{W_laxpair_x}
\begin{cases}
\ii W^{-}_{2,x}+(k-\mu)W_2^--\nu_1e^{-k\delta}W_1^+=uW_2^-, \\
\ii W^{+}_{2,x}+(k-\mu)W_2^+-\nu_2e^{k\delta}W_1^-=-vW_2^+.
\end{cases}
\end{align}
We define a solution $(W_1,W_2)=(M_1,M_2)$ to \eqref{W_laxpair} by the asymptotic behavior
\begin{equation}
M_1\sim 1, \qquad M_2\sim 1,\qquad x\to -\infty.
\end{equation}
The existence of this solution implies the conditions
\begin{equation}
k-\mu=\nu_1 e^{-k\delta}=\nu_2 e^{k\delta}.
\end{equation}
We assume these conditions hold and define $\zeta(k)=k-\mu(k)$, so that the full Lax pair in terms of the particular eigenfunctions $M_1$ and $M_2$ is
\begin{subequations}\label{W_laxpair}
\begin{align}
&\ii M^{-}_{2,x}+\zeta(M_2^--M_1^+)=uM_2^-,  \\
&\ii M^{+}_{2,x}+\zeta(M_2^+-M_1^-)=-vM_2^+, \\
&M_{1,t}^{+}+\ii M_{1,xx}^{+}+2\zeta M_{1,x}^+-(\ii Tu_x+\ii\tilde{T}v_x-u_x)M_1^+=0, \\
&M_{1,t}^{-}+\ii M_{1,xx}^{-}+2\zeta M_{1,x}^--(\ii Tv_x+\ii\tilde{T}u_x+v_x)M_1^-=0, \\
&M_{2,t}^{+}+\ii M_{2,xx}^{+}+2\zeta M_{2,x}^+-(\ii Tv_x+\ii\tilde{T}u_x-v_x)M_1^+=0, \\
&M_{2,t}^{-}+\ii M_{2,xx}^{-}+2\zeta M_{2,x}^--(\ii Tu_x+\ii\tilde{T}v_x+u_x)M_2^-=0. 
\end{align}
\end{subequations}

Define
\begin{equation}\label{sigmadefinition}
\sigma_1 \coloneqq \log \frac{M_2^-}{M_1^+}, \qquad \sigma_2\coloneqq \log \frac{M_1^-}{M_2^+}.
\end{equation}
Then $\sigma_1$ and $\sigma_2$ decay to zero as $|x| \to \infty$, and, for large enough $\zeta$, $M_1^+$ and $M_2^+$ are nonzero. 
Thus the function
\begin{align}\label{Adef}
A(z) \coloneqq \begin{cases}  \frac{M_{1,x}}{M_1}(z - \ii\delta/2), & 0 < \im z < \delta, \\
\frac{M_{2,x}}{M_2}(z - 3\ii\delta/2), & \delta < \im z < 2\delta, 
\end{cases}
\end{align}
is analytic on $C \setminus (C_0 \cup C_\delta)$, where $C$ is the cylinder $C = \C / 2\delta \ii \Z$ and $C_0 \cong \R$ and $C_\delta \cong \ii\delta + \R$.
The function $A(z)$ defined in \eqref{Adef} satisfies the jump conditions
$$A_+(z) - A_-(z) = \begin{cases}
J_0(x) = \frac{M_{1,x}^+}{M_1^+}(x) - \frac{M_{2,x}^-}{M_2^-}(x) = - \sigma_{1,x}, &z=x \in \R,
	\\
J_1(x) \coloneqq \frac{M_{2,x}^+}{M_2^+}(x) - \frac{M_{1,x}^-}{M_1^-}(x) = - \sigma_{2,x}, &z = x +\ii\delta \in \R + \ii\delta.
\end{cases}
$$
Lemma \ref{RHlemma} therefore gives, for $\zeta$ sufficiently large,
$$A^\pm(z) = \begin{cases}
-\frac{(T\sigma_{1,x})(x) + (\tilde{T}\sigma_{2,x})(x)}{2\ii} \mp \frac{1}{2} \sigma_{1,x}(x), & z = x \in C_0 \cong \R,
	\\
-\frac{(T\sigma_{2,x})(x) + (\tilde{T}\sigma_{1,x})(x)}{2\ii} \mp \frac{1}{2} \sigma_{2,x}(x), & z = x + \ii \delta \in C_\delta \cong \R + \ii \delta.
\end{cases}$$
That is,
\begin{subequations}\label{W12sigma12}
\begin{align}
  \frac{M_{1,x}^+}{M_1^+} = -\frac{(T\sigma_{1,x}) + (\tilde{T}\sigma_{2,x})}{2\ii} - \frac{1}{2} \sigma_{1,x},
\qquad  \frac{M_{2,x}^-}{M_2^-} = -\frac{(T\sigma_{1,x}) + (\tilde{T}\sigma_{2,x})}{2\ii} + \frac{1}{2} \sigma_{1,x},
  	\\
  \frac{M_{2,x}^+}{M_2^+} = -\frac{(T\sigma_{2,x}) + (\tilde{T}\sigma_{1,x})}{2\ii} - \frac{1}{2} \sigma_{2,x},
\qquad
 \frac{M_{1,x}^-}{M_1^-} = -\frac{(T\sigma_{2,x}) + (\tilde{T}\sigma_{1,x})}{2\ii} + \frac{1}{2} \sigma_{2,x}.
\end{align}	
  \end{subequations}

We next use these expressions to express the Lax pair \eqref{W_laxpair} in terms of $\sigma_1$ and $\sigma_2$. Dividing by $M_2^-$ and $M_2^+$ in the first and second equation, respectively, we find
\begin{align*}
& \ii\frac{M_{2,x}^-}{M_2^-} + \zeta(k) (1 - e^{-\sigma_1}) = u,
	\\ 
& \ii\frac{M_{2,x}^+}{M_2^+} + \zeta(k) (1 - e^{\sigma_2}) = -v.
\end{align*}
Using \eqref{W12sigma12} and \eqref{Ppm}, this can be written as
\begin{subequations}\label{sigmaxpart}
\begin{align}
&-\ii P_+\sigma_{1,x}-\frac12 \tilde{T}\sigma_{2,x}+\zeta(1-e^{-\sigma_1})=u,
	\\ 
&-\ii P_- \sigma_{2,x} -\frac12\tilde{T}\sigma_{1,x}+\zeta(1-e^{\sigma_2})=-v.
\end{align}
\end{subequations}

The next lemma provides the time evolution of $\sigma_1$ and $\sigma_2$.

\begin{lemma}
 The functions $\sigma_1$ and $\sigma_2$ satisfy
  \begin{align*}
    \sigma_{1,t} - \ii\sigma_{1,xx} + \sigma_{1,x}(T\sigma_{1,x} + \tilde{T} \sigma_{2,x}) + 2\zeta \sigma_{2,x} - 2u_x = 0,
    	\\   
    \sigma_{2,t} - \ii\sigma_{2,xx} + \sigma_{2,x}(T\sigma_{2,x} + \tilde{T} \sigma_{1,x}) + 2\zeta \sigma_{1,x} - 2v_x = 0.
  \end{align*}
In particular, 
\begin{align}\label{sigma1plussigma2conserved}
\frac{\mathrm{d}}{\mathrm{d}t} \int_{\R} (\sigma_1 + \sigma_2) \,\mathrm{d}x = 0.
\end{align}
\end{lemma}
\begin{proof}
Using that $\zeta = k- \mu$, the $t$-part of \eqref{W_laxpair} can be written as
\begin{subequations}\label{bidirectionaltpart2}
\begin{align}
&\frac{M_{1,t}^+}{M_1^+} + \ii\frac{M_{1,xx}^+}{M_1^+} 
+2\zeta \frac{M_{1,x}^+}{M_1^+} 
-(\ii Tu_x + \ii\tilde{T}v_x - u_x ) = 0, 
	 \\ 
&\frac{M_{1,t}^-}{M_{1}^-} + \ii\frac{M_{1,xx}^-}{M_{1}^-} 
+2\zeta \frac{M_{1,x}^-}{M_{1}^-}
-(\ii Tv_x  + \ii\tilde{T}u_x  + v_x) = 0, 
	 \\ 
&\frac{M_{2,t}^+}{M_2^+} + \ii\frac{M_{2,xx}^+}{M_2^+} 
+2\zeta \frac{M_{2,x}^+}{M_2^+} 
- (\ii Tv_x  + \ii\tilde{T}u_x  - v_x ) = 0, 
	 \\ 
&\frac{M_{2,t}^-}{M_2^-} + \ii\frac{M_{2,xx}^-}{M_2^-} 
+2\zeta \frac{M_{2,x}^-}{M_2^-} 
-(\ii Tu_x + \ii \tilde{T}v_x + u_x) =0. 
\end{align}
\end{subequations}

Subtracting the first from the fourth equation and using that
$$\sigma_{1,t} = \frac{M_{2,t}^-}{M_2^-} - \frac{M_{1,t}^+}{M_1^+}, \qquad
\sigma_{1,x} = \frac{M_{2,x}^-}{M_2^-} - \frac{M_{1,x}^+}{M_1^+},$$
we obtain
\begin{align}\label{sigma1tpartW}
\sigma_{1,t} - \ii\bigg(\frac{M_{2,xx}^-}{M_2^-} - \frac{M_{1,xx}^+}{M_1^+} \bigg)
+2\zeta \sigma_{1,x}
- 2u_x = 0.
\end{align}
Similarly, subtracting the third from the second equation and using that
$$\sigma_{2,t} = \frac{M_{1,t}^-}{M_1^-} - \frac{M_{2,t}^+}{M_2^+}, \qquad
\sigma_{2,x} = \frac{M_{1,x}^-}{M_1^-} - \frac{M_{2,x}^+}{M_2^+},$$
we obtain
\begin{align}\label{sigma2tpartW}
\sigma_{2,t} - \ii\bigg(\frac{M_{1,xx}^-}{M_{1}^-} - \frac{M_{2,xx}^+}{M_2^+} \bigg)
+2\zeta \sigma_{2,x}
- 2v_x = 0, 
\end{align}

Next note that
$$\sigma_{2,xx} = \frac{M_{1,xx}^-}{M_1^-} - \frac{M_{2,xx}^+}{M_2^+} - \bigg(\frac{M_{1,x}^-}{M_1^-}\bigg)^2 + \bigg(\frac{M_{2,x}^+}{M_2^+}\bigg)^2$$
and, by \eqref{W12sigma12},
\begin{align*}
\bigg(\frac{M_{1,x}^-}{M_1^-}\bigg)^2 - \bigg(\frac{M_{2,x}^+}{M_2^+}\bigg)^2
& = \bigg(\frac{M_{1,x}^-}{M_1^-} - \frac{M_{2,x}^+}{M_2^+}\bigg)
\bigg(\frac{M_{1,x}^-}{M_1^-} + \frac{M_{2,x}^+}{M_2^+}\bigg)
	\\
& = \sigma_{2,x}
\bigg(-\frac{(T\sigma_{2,x}) + (\tilde{T}\sigma_{1,x})}{2i} + \frac{1}{2} \sigma_{2,x}
-\frac{(T\sigma_{2,x}) + (\tilde{T}\sigma_{1,x})}{2\ii} - \frac{1}{2} \sigma_{2,x}\bigg)
	\\
& = \ii \sigma_{2,x}(T\sigma_{2,x} + \tilde{T}\sigma_{1,x}).
\end{align*}
Hence
$$\frac{M_{1,xx}^-}{M_{1}^-} - \frac{M_{2,xx}^+}{M_2^+}
= \sigma_{2,xx} + \ii \sigma_{2,x}(T\sigma_{2,x} + \tilde{T}\sigma_{1,x}).$$ 
Similarly,
$$\frac{M_{2,xx}^-}{M_2^-} - \frac{M_{1,xx}^+}{M_1^+}
= \sigma_{1,xx} + \ii \sigma_{1,x}(T\sigma_{1,x} + \tilde{T}\sigma_{2,x}).$$ 
Thus, the stated equations for $\sigma_{1,t}$ and $\sigma_{2,t}$ follow from \eqref{sigma1tpartW} and \eqref{sigma2tpartW}.
Equation \eqref{sigma1plussigma2conserved} follows because $T$ and $\tilde{T}$ are anti-self-adjoint, see \eqref{anti_self_adjointTTtilde}. 
\end{proof}

We can use the fact that $\int_{\R} (\sigma_1 + \sigma_2)\,\mathrm{d}x$ is conserved for all $\zeta$ to determine an infinite sequence of conservation laws for the non-chiral ILW equation \eqref{2ilw}. 

As $\zeta \to \infty$, we have the expansions
\begin{equation}\label{sigmaexpansions}\sigma_1(x,t;k) = \sum_{n=1}^\infty \frac{\sigma_1^{(n)}(x,t)}{\zeta(k)^n},  \qquad
\sigma_2(x,t;k) = \sum_{n=1}^\infty \frac{\sigma_2^{(n)}(x,t)}{\zeta(k)^n}.
\end{equation}
It follows from \eqref{sigma1plussigma2conserved} that
$$I_n = \int_\R (\sigma_1^{(n)} + \sigma_2^{(n)}) \,\mathrm{d}x, \qquad n \in \mathbb{N},$$
forms an infinite sequence of conserved quantities. Substituting \eqref{sigmaexpansions} into \eqref{sigmaxpart} gives
\begin{align}\label{sigma1sigma2implicit}
&u=\zeta\bigg(1-\exp \bigg(-\sum_{n=1}^{\infty}  \sigma_1^{(n)} \zeta^{-n}  \bigg)     \bigg)-\ii P_+ \sum_{n=1}^{\infty} \sigma_{1,x}^{(n)}\zeta^{-n}-\frac12\tilde{T} \sum_{n=1}^{\infty} \sigma_{2,x}^{(n)}\zeta^{-n}, \\
&v=-\zeta\bigg(1-\exp \bigg(\sum_{n=1}^{\infty}  \sigma_2^{(n)} \zeta^{-n}  \bigg)     \bigg)+\ii P_- \sum_{n=1}^{\infty} \sigma_{2,x}^{(n)}\zeta^{-n}+\frac12\tilde{T} \sum_{n=1}^{\infty} \sigma_{1,x}^{(n)}\zeta^{-n}.
\end{align}
Using $P_{\pm}^*=-P_{\mp}$, we see that \eqref{sigma1sigma2implicit} is precisely the complex conjugate of \eqref{WZimplicit} with the identifications $\zeta=1/\epsilon^*$, $\sigma_1=W^*$, $\sigma_2=Z^*$. Thus, the remainder of the proof is similiar to the proof of the first four conservation laws in Subsection \ref{conservationslawsBacklundproof} and hence omitted. 

\subsection{Direct verification of first three conservation laws}\label{directverification}
To verify the first conservation law \eqref{conservation_law1}, we only need the non-chiral ILW equation \eqref{2ilw}. For $I_{1,u}$, 
\begin{align*}
\frac{\mathrm{d}I_{1,u}}{\mathrm{d}t}=&\int_\R u_t\,\mathrm{d}x \\
=&\int_\R (-2uu_x-Tu_{xx}-\tilde{T}v_{xx})\,\mathrm{d}x \\
=&\left[-u^2-Tu_{x}-\tilde{T}v_x\right]^{\infty}_{-\infty}=0.
\end{align*}
The verification for $I_{1,v}$ is similar.

To verify the second conservation law \eqref{conservation_law2}, we also need the anti-self-adjointness \eqref{anti_self_adjointTTtilde} of the operators $T$ and $\tilde{T}$ from Proposition \ref{TpropertiesR}:
\begin{align*}
\frac{\mathrm{d}I_2}{\mathrm{d}t}=&\int_\R (uu_t-vv_t)\,\mathrm{d}x \\
=&-\int_\R (2u^2u_x+uTu_{xx}+v\tilde{T}v_{xx}+2v^2v_x+vTv_{xx}+v\tilde{T}u_{xx})\,\mathrm{d}x \\
=&\int_\R (u_x Tu_x+u_x\tilde{T}v_x+v_xTv_x+v_x\tilde{T}u_x)\,\mathrm{d}x=0.
\end{align*}

To verify the third conservation law \eqref{conservation_law3}, we need the identity $(\tilde{T}Tf)(x)=(T\tilde{T}f)(x)$ from Proposition \ref{TpropertiesR}:
\begin{align*}
 \frac{\mathrm{d}I_3}{\mathrm{d}t}
 = &\; \frac{1}{2} \int_\R \big(u_t T u_x +u T u_{xt} + v_t T v_x + v T v_{xt} + u_t\tilde{T}v_x+ u\tilde{T}v_{xt} + v_t\tilde{T} u_x +  v\tilde{T} u_{xt}\big) \mathrm{d}x \\
 &+ \int_\R (u^2u_t + v^2v_t)\,\mathrm{d}x 
	\\
=& \int_\R \big(u_t T u_x  + v_t T v_x + u_t\tilde{T}v_x+ v_t\tilde{T} u_x\big) \mathrm{d}x \\	
&+ \int_\R \big(-2u^3u_x - u^2 Tu_{xx} -u^2 \tilde{T}v_{xx} + 2v^3v_x + v^2Tv_{xx} + v^2\tilde{T}u_{xx}\big)\mathrm{d}x 
	\\
= &  \int_\R \bigg\{-(2uu_x + Tu_{xx} + \tilde{T}v_{xx}) T u_x  + (2vv_x + Tv_{xx} + \tilde{T}u_{xx}) T v_x 
	\\
& \qquad\qquad\,-(2uu_x + Tu_{xx} + \tilde{T}v_{xx})\tilde{T}v_x+ (2vv_x + Tv_{xx} + \tilde{T}u_{xx})\tilde{T} u_x\bigg\} \mathrm{d}x \\
&+\int_\R \big(2uu_x Tu_{x} + 2uu_x \tilde{T}v_{x} - 2vv_xTv_{x} - 2vv_x\tilde{T}u_{x}\big)\mathrm{d}x \\
= & \int_\R \big(-(\tilde{T}v_{xx}) T u_x + (\tilde{T}u_{xx}) T v_x 
-(Tu_{xx})\tilde{T}v_x + (Tv_{xx} )\tilde{T} u_x\big) \mathrm{d}x
	\\
= & \int_\R \big( u_x T\tilde{T}v_{xx} - u_{xx} \tilde{T} T v_x 
+ u_{xx} T \tilde{T}v_x - u_x \tilde{T}Tv_{xx}\big) \mathrm{d}x \\
= & \int_\R \big( u_x T\tilde{T}v_{xx}  - u_{xx} (T\tilde{T} v_x + 2i\tilde{T}v_x) 
+ u_{xx} T \tilde{T}v_x - u_x (T\tilde{T}v_{xx} + 2i\tilde{T}v_{xx})\big) \mathrm{d}x
	\\
= & \int_\R \big(-u_{xx} (2\ii\tilde{T}v_x) - u_x (2\ii\tilde{T}v_{xx})\big) \mathrm{d}x=0.
\end{align*}

\section{Discussion}
In this paper we have presented a Lax pair, Hirota bilinear form, B\"{a}cklund transformation, and an infinite sequence of conservation laws for the non-chiral ILW equation \eqref{2ilw}. 
 While our results are generalizations of those for the standard ILW equation \cite{kodama1981,satsuma1979,matsuno1979}, we emphasize that the non-chiral ILW equation does not contain the standard ILW equation as a limiting case and exhibits features not present in the single-component case \cite{berntson2020a}. We conclude by mentioning four directions for future research. 
\begin{enumerate}
\item The initial value problem for the standard ILW equation can be solved via an inverse scattering transform \cite{kodama1981}. The Lax pair \eqref{2ilw_laxpair} would expectedly provide a starting point for an analogous inverse scattering transform to solve the initial value problem for the ncILW equation \eqref{2ilw}. Development of such a method is a natural problem for future work. Preliminary calculations indicate that there are technical challenges not present in the standard case.  

\item The derivation of the Lax pair in Section \ref{laxsec} suggests that an $N$-component generalization of the non-chiral ILW equation \eqref{2ilw} could be obtained starting from a RH problem with $N$ jumps on the cylinder. It would be particularly interesting to investigate the $N\to \infty$ limit of this construction.
\item The cubic Szeg\"{o} equation \cite{gerard2010} and the half-waves map \cite{zhou2015,lenzmann2018,berntson2020d} are two recently introduced equations that, like the non-chiral ILW equation \eqref{2ilw}, have nonlocalities given by a Fourier multiplier. Both of these equations are integrable by virtue of Lax representations, admit $N$-soliton solutions obtained via pole ans\"atze, and possess an infinite number of conservation laws. It would be interesting to investigate these equations from the perspective taken in this paper by constructing their Hirota forms and B\"{a}cklund transformations. 
\item The INLS equation discussed in Section~\ref{subsec1.3} is obtained via a multi-scale expansion of the standard ILW equation \cite{pelinovsky1995}. 
It would be interesting to apply the same technique to the non-chiral ILW equation \eqref{2ilw} in search of a non-chiral intermediate INLS equation. 
\end{enumerate}

\appendix

\section{Properties of the $T$ and $\tilde{T}$ operators}
\label{app:TT}
In this section we collect and prove several identities for the $T$ and $\tilde{T}$ operators in \eqref{TT}.

We note that a necessary condition for the existence of any of the products of transforms $TTf$, $\tilde{T}\tilde{T}f$, $T\tilde{T}f$, $\tilde{T}Tf$ on the line is $\int_\R f(x)\,\mathrm{d}x=0$. 

\begin{proposition}\label{TpropertiesR}
The operators $T$ and $\tilde{T}$ have the following properties
\begin{align}
\label{Tderivative}
& \partial_x (Tf)(x)=(Tf')(x),\quad \partial_x (\tilde{T}f)(x)=(\tilde{T}f')(x),   \\
\label{anti_self_adjointTTtilde}
&\int_\R f (Tg)\,\mathrm{d}x=-\int_\R (Tf)g\,\mathrm{d}x,\qquad \int_\R f (\tilde{T}g)\,\mathrm{d}x=-\int_\R (\tilde{T}f)g\,\mathrm{d}x, \\
\label{Tcommutator}
&(\tilde{T} T f)(x) = (T\tilde{T}f)(x) , \qquad x\in \R, \\
\label{TTcommutator} 
&(\tilde{T}\tilde{T} f)(x) = (TTf)(x)+f(x),\qquad x\in \R, \\
\label{TplusTminus}
&((T+\tilde{T})(T-\tilde{T})f)(x)=-f(x),\qquad x\in \R.
\end{align}
\end{proposition}

\begin{proof}
\eqref{Tderivative}.
By the definition of $T$,
\begin{align*}
\partial_x(Tf)(x)=&\partial_x \bigg(\frac1{2\delta}\pvint_{\R} \tanh\bigg(\frac{\pi(x'-x)}{2\delta}\bigg)f(x')\,\mathrm{d}x'\bigg) \\
=&\frac1{2\delta}\partial_x \bigg(\lim_{\epsilon\to 0} \bigg(  \int_{-\infty}^{-\epsilon}+\int_{\epsilon}^{\infty}  \bigg) \tanh\bigg(\frac{\pi(x'-x)}{2\delta}\bigg)f(x')\,\mathrm{d}x'\bigg) \\
=&-\frac1{2\delta}\lim_{\epsilon\to 0} \bigg(  \int_{-\infty}^{-\epsilon}+\int_{\epsilon}^{\infty}  \bigg) \partial_{x'}\bigg(\tanh\bigg(\frac{\pi(x'-x)}{2\delta}\bigg)\bigg)f(x')\,\mathrm{d}x'\bigg)  \\
=&-\frac1{2\delta} \bigg[\tanh\bigg(\frac{\pi(x'-x)}{2\delta}\bigg)f(x')\bigg]^{x'=-\infty}_{x'=-\epsilon}-\frac1{2\delta} \bigg[\tanh\bigg(\frac{\pi(x'-x)}{2\delta}\bigg)f(x')\bigg]^{x'=\epsilon}_{x'=\infty}\\
&+\frac{1}{2\delta}\lim_{\epsilon\to 0}\bigg(  \int_{-\infty}^{-\epsilon}+\int_{\epsilon}^{\infty}  \bigg) \tanh\bigg(\frac{\pi(x'-x)}{2\delta}\bigg)f'(x')\,\mathrm{d}x' \\
=&\frac1{2\delta}\pvint_{\R} \tanh\bigg(\frac{\pi(x'-x)}{2\delta}\bigg)f'(x')\,\mathrm{d}x'.
\end{align*}
The proof of the corresponding formula for $\tilde{T}$ is similar.
\hfill\break

\eqref{anti_self_adjointTTtilde}.
By the definition of $T$, 
\begin{align*}
\int_\R f(Tg)\,\mathrm{d}x=&\int_\R f(x)\left( \frac1{2\delta} \pvint_\R \coth\bigg(\frac{\pi(x'-x)}{2\delta}\bigg)  g(x')\,\mathrm{d}x'    \right)\mathrm{d}x \\
=&\lim_{\epsilon\downarrow 0} \int_\R f(x)\frac1{2\delta}\left( \int_{-\infty}^{x-\epsilon}+\int_{x+\epsilon}^{\infty}    \right) \coth \bigg(\frac{\pi(x'-x)}{2\delta}\bigg)g(x')\,\mathrm{d}x'\,\mathrm{d}x \\
=&\lim_{\epsilon\downarrow 0} \int_\R g(x')\frac1{2\delta}\left( \int_{-\infty}^{x'-\epsilon}+\int_{x'+\epsilon}^{\infty}    \right) \coth\bigg(\frac{\pi(x'-x)}{2\delta}\bigg)f(x)\,\mathrm{d}x\,\mathrm{d}x' \\
=& -\int_\R g(x')(Tf)(x')\,\mathrm{d}x'.
\end{align*}
Similarly, by the definition of $\tilde{T}$,
\begin{align*}
\int_\R f(\tilde{T}g)\,\mathrm{d}x=&\int_\R f(x)\left( \frac1{2\delta} \int_\R \tanh\bigg(\frac{\pi(x'-x)}{2\delta}\bigg)  g(x')\,\mathrm{d}x'    \right)\mathrm{d}x \\
=& \int_\R f(x)\frac{1}{2\delta} \int_\R \tanh\bigg(\frac{\pi(x'-x)}{2\delta}\bigg)g(x')\,\mathrm{d}x'\,\mathrm{d}x \\
=& \int_\R g(x')\frac{1}{2\delta} \int_\R \tanh\bigg(\frac{\pi(x'-x)}{2\delta}\bigg)f(x)\,\mathrm{d}x\,\mathrm{d}x' \\
=& -\int_\R g(x')(\tilde{T}f)(x')\,\mathrm{d}x'.
\end{align*}
\hfill\break

\eqref{Tcommutator}.
By the definitions of $T$ and $\tilde{T}$,
\begin{align*}
(\tilde{T} T f)(x) 
= &\; 
 \frac{1}{2\delta} \int_\R \tanh\bigg(\frac{\pi(x'-x)}{2\delta}\bigg) 
 \frac{1}{2\delta} \pvint \coth\bigg(\frac{\pi(z-x')}{2\delta}\bigg) f(z)\, \mathrm{d}z\, \mathrm{d}x'
	\\
= &\; 
 \frac{1}{(2\delta)^2} \int_\R \tanh\bigg(\frac{\pi(x'-x)}{2\delta}\bigg) 
\bigg\{\int_{\R +\ii 0} \coth\bigg(\frac{\pi(z-x')}{2\delta}\bigg) f(z) \mathrm{d}z + 2 \ii \delta f(x')\bigg\} \,\mathrm{d}x'
	\\
= &\; 
 \frac{1}{(2\delta)^2} \int_\R \tanh\bigg(\frac{\pi(x'-x)}{2\delta}\bigg) 
\int_{\R + \ii 0} \coth\bigg(\frac{\pi(z-x')}{2\delta}\bigg) f(z)\, \mathrm{d}z\,  \mathrm{d}x'
+  \ii (\tilde{T}f)(x).
\end{align*}
Deforming the $x'$-contour downward, we find
\begin{align*}
(\tilde{T} T f)(x) 
=&\;  \frac{1}{(2\delta)^2} \int_{\R-\ii\delta+\ii 0} \tanh\bigg(\frac{\pi(x'-x)}{2\delta}\bigg) 
 \int_{\R + \ii 0} \coth\bigg(\frac{\pi(z-x')}{2\delta}\bigg) f(z)\, \mathrm{d}z\,\mathrm{d}x'
	\\
& + E +  \ii (\tilde{T}f)(x),
\end{align*}
where
\begin{align*}
E(x) \coloneqq  \frac{1}{(2\delta)^2} \lim_{R \to \infty} \bigg(\int_{-R}^{-R-\ii\delta} + \int_{R-\ii\delta}^{R} \bigg)\tanh\bigg(\frac{\pi(x'-x)}{2\delta}\bigg) 
 \int_{\R + \ii 0} \coth\bigg(\frac{\pi(z-x')}{2\delta}\bigg) f(z)\, \mathrm{d}z\,\mathrm{d}x'
\end{align*}
In fact, $E$ vanishes because
\begin{align*}
E(x) = &\; \frac{1}{(2\delta)^2} \lim_{R \to \infty} \bigg\{\int_{-R}^{-R-\ii\delta} (-1)
 \int_{\R + \ii 0} f(z)\, \mathrm{d}z\,  \mathrm{d}x'
 + \int_{R-\ii\delta}^{R}   \int_{\R + \ii 0} (-1) f(z)\, \mathrm{d}z\,  \mathrm{d}x'\bigg\}
	\\
= &\; - \frac{1}{(2\delta)^2} \lim_{R \to \infty} \bigg\{\int_{-R}^{-R-\ii\delta} \mathrm{d}x' + \int_{R-\ii\delta}^{R}  \mathrm{d}x'\bigg\} \int_{\R} f(z) \,\mathrm{d}z
	\\
= &\; - \frac{1}{(2\delta)^2} \lim_{R \to \infty} \big\{-\ii\delta + \ii\delta \big\} \int_{\R} f(z) \,\mathrm{d}z
= 0.
\end{align*}
Thus we find after the change of variables $y = x' + \ii\delta$ that
\begin{align*}
(\tilde{T} T f)(x) 
=&\;  \frac{1}{(2\delta)^2} \int_{\R+\ii 0} \coth\bigg(\frac{\pi(y-x)}{2\delta}\bigg) 
 \int_{\R + \ii 0} \tanh\bigg(\frac{\pi(z-y)}{2\delta}\bigg) f(z) \mathrm{d}z\, \mathrm{d}y
+  \ii (\tilde{T}f)(x).
\end{align*}
Thus, 
\begin{align*}
(\tilde{T} T f)(x) 
=&\;  \frac{1}{(2\delta)^2} \int_{\R+\ii 0} \coth\bigg(\frac{\pi(y'-x)}{2\delta}\bigg) 
 \int_{\R} \tanh\bigg(\frac{\pi(z-y)}{2\delta}\bigg) f(z)\, \mathrm{d}z\,  \mathrm{d}y
 +  \ii (\tilde{T}f)(x),
	\\
=&\;  \frac{1}{2\delta} \int_{\R+\ii 0} \coth\bigg(\frac{\pi(y-x)}{2\delta}\bigg) (\tilde{T}f)(y)\,  \mathrm{d}y
 +  \ii (\tilde{T}f)(x)
 	\\
=&\;  \frac{1}{2\delta} \pvint_{\R} \coth\bigg(\frac{\pi(y-x)}{2\delta}\bigg) (\tilde{T}f)(y)\,  \mathrm{d}y
= (T\tilde{T}f)(x) .
\end{align*}
\hfill\break

\eqref{TTcommutator}.
By the definition of $\tilde{T}$
\begin{align*}
(\tilde{T}\tilde{T}f)(x)=&\frac{1}{2\delta}\int_\R  \tanh\bigg(\frac{\pi(x'-x)}{2\delta}\bigg)\frac{1}{2\delta}\int_\R  \tanh\bigg(\frac{\pi(z-x')}{2\delta}\bigg) f(z)\,\mathrm{d}z\,\mathrm{d}x' \\
=&\frac{1}{(2\delta)^2}\int_{\R-\ii 0}  \tanh\bigg(\frac{\pi(x'-x)}{2\delta}\bigg)\int_{\R}  \tanh\bigg(\frac{\pi(z-x')}{2\delta}\bigg) f(z)\,\mathrm{d}z\,\mathrm{d}x' 
\end{align*}
Deforming the $x'$-contour downward, we find
\begin{align*}
(\tilde{T}\tilde{T}f)(x)=\frac1{(2\delta)^2}\int_{\R-\ii \delta+\ii 0}\tanh\bigg(\frac{\pi(x'-x)}{2\delta}\bigg)\int_{\R+\ii 0}  \tanh\bigg(\frac{\pi(z-x')}{2\delta}\bigg) f(z)\,\mathrm{d}z\,\mathrm{d}x'  + E,
\end{align*}
where 
\begin{align*}
E(x) \coloneqq  \frac{1}{(2\delta)^2} \lim_{R \to \infty} \bigg(\int_{-R}^{-R-\ii\delta} + \int_{R-\ii\delta}^{R} \bigg)\tanh\bigg(\frac{\pi(x'-x)}{2\delta}\bigg) 
 \int_{\R + \ii 0} \tanh\bigg(\frac{\pi(z-x')}{2\delta}\bigg) f(z)\, \mathrm{d}z\,\mathrm{d}x'.
\end{align*}
In fact, $E$ vanishes because
\begin{align*}
E(x) = &\; \frac{1}{(2\delta)^2} \lim_{R \to \infty} \bigg\{\int_{-R}^{-R-\ii\delta} (-1)
 \int_{\R + \ii 0} f(z)\, \mathrm{d}z\,  \mathrm{d}x'
 + \int_{R-\ii\delta}^{R}   \int_{\R + \ii 0} (-1) f(z)\, \mathrm{d}z\,  \mathrm{d}x'\bigg\}
	\\
= &\; - \frac{1}{(2\delta)^2} \lim_{R \to \infty} \bigg\{\int_{-R}^{-R-\ii\delta} \mathrm{d}x' + \int_{R-\ii\delta}^{R}  \mathrm{d}x'\bigg\} \int_{\R} f(z) \,\mathrm{d}z
	\\
= &\; - \frac{1}{(2\delta)^2} \lim_{R \to \infty} \big\{-\ii\delta + \ii\delta \big\} \int_{\R} f(z) \,\mathrm{d}z
= 0.
\end{align*}
Thus we find after the change of variables $y=x'+\ii\delta$ that 
\begin{align*}
(\tilde{T}\tilde{T}f)(x)=&\frac{1}{(2\delta)^2}\int_{\R+\ii 0}\coth\bigg(\frac{\pi(y-x)}{2\delta}\bigg)\int_{\R-\ii 0}  \coth\bigg(\frac{\pi(z-y)}{2\delta}\bigg) f(z)\,\mathrm{d}z\,\mathrm{d}y.
\end{align*}
Thus,
\begin{align*}
(\tilde{T}\tilde{T} f)(x)=&\frac{1}{2\delta}\int_{\R+\ii 0}\coth\bigg(\frac{\pi(y-x)}{2\delta}\bigg)\bigg\{\frac{1}{2\delta}\pvint_\R \coth \bigg(\frac{\pi(z-y)}{2\delta}\bigg)f(z)\,\mathrm{d}z+ \ii f(y) \bigg\}\,\mathrm{d}y \\
=&\frac{1}{2\delta}\int_{\R+\ii 0}\coth\bigg(\frac{\pi(y-x)}{2\delta}\bigg)\big\{ (Tf)(y)+ \ii f(y) \big\}\,\mathrm{d}y  \\
=&\frac{1}{2\delta}\pvint_{\R} \coth\bigg(\frac{\pi(y-x)}{2\delta}\bigg)\big\{ (Tf)(y)+ \ii f(y) \big\}\,\mathrm{d}y-\ii \{ (Tf)(x)+ \ii f(x) \big\} \\
=& (TT)(f)+\ii (Tf)(x)-\ii(Tf)(x)+f(x)=(TTf)(x)+f(x).
\end{align*}
\hfill\break

\eqref{TplusTminus}.
In the case where $Tf$ and $\tilde{T}f$ exist separately, \eqref{Tcommutator} and \eqref{TTcommutator} immediately imply \eqref{TplusTminus}. We show that \eqref{TplusTminus} also holds when only $(T-\tilde{T})f$ exists. By the definitions of $T$ and $\tilde{T}$ and the identity
\begin{equation}\label{hyperbolicidentity}
\csch\,2z=\frac12(\coth z-\tanh z),
\end{equation}
\begin{align*}
(\tilde{T}(T-\tilde{T})f)(x)=&\frac1{2\delta^2}\int_\R \tanh\bigg(\frac{\pi(x'-x)}{2\delta}\bigg)\pvint_\R \csch\bigg(\frac{\pi(z-x')}{\delta}\bigg)f(z)\,\mathrm{d}z\,\mathrm{d}x'\\
=&\frac1{2\delta^2}\int_\R \tanh\bigg(\frac{\pi(x'-x)}{2\delta}\bigg)\bigg\{\int_{\R+\ii 0} \csch\bigg(\frac{\pi(z-x')}{\delta}\bigg)f(z)\,\mathrm{d}z+\ii\delta f(x')\bigg\}\,\mathrm{d}x' \\
=&\frac1{2\delta^2}\int_\R \tanh\bigg(\frac{\pi(x'-x)}{2\delta}\bigg)\int_{\R+\ii 0} \csch\bigg(\frac{\pi(z-x')}{\delta}\bigg)f(z)\,\mathrm{d}z\,\mathrm{d}x' +\ii(\tilde{T}f)(x)
\end{align*}
Deforming the $x'$-contour downwards, we find
\begin{align*}
(\tilde{T}(T-\tilde{T})f)(x)=&\frac1{2\delta^2}\int_{\R-\ii\delta+\ii 0} \tanh\bigg(\frac{\pi(x'-x)}{2\delta}\bigg)\int_{\R+\ii 0} \csch\bigg(\frac{\pi(z-x')}{\delta}\bigg)f(z)\,\mathrm{d}z\,\mathrm{d}x' \\
&+\ii(Tf)(x)+E(x),
\end{align*}
where
\begin{align*}
E(x)\coloneqq \frac{1}{2\delta^2} \lim_{R\to\infty} \bigg( \int_{-R}^{-R+\ii\delta}+\int_{R-\ii\delta}^R\bigg) \tanh\bigg(\frac{\pi(x'-x)}{2\delta}\bigg)\pvint_\R \csch\bigg(\frac{\pi(z-x')}{\delta}\bigg)f(z)\,\mathrm{d}z\,\mathrm{d}x',
\end{align*}
but $E$ vanishes because $\csch\, x=\mathrm{O}(e^{-|x|})$ as $|x|\to\infty$. 
Thus we find after a change of variables $y=x'+\ii\delta$ that
\begin{align*}
(\tilde{T}(T-\tilde{T})f)(x)=&-\frac1{2\delta^2}\int_{\R+\ii 0} \coth\bigg(\frac{\pi(y-x)}{2\delta}\bigg)\int_{\R-\ii 0} \csch\bigg(\frac{\pi(z-y)}{\delta}\bigg)f(z)\,\mathrm{d}z\,\mathrm{d}y\\
&+\ii (\tilde{T}f)(x) \\
=&-\frac1{2\delta^2}\int_{\R+\ii 0} \coth\bigg(\frac{\pi(y-x)}{2\delta}\bigg)\bigg\{\pvint_{\R} \csch\bigg(\frac{\pi(z-y)}{\delta}\bigg)f(z)\,\mathrm{d}z+\ii\delta f(z)\bigg\}\mathrm{d}y\\ 
&+\ii(\tilde{T}f)(x) \\
=&-\frac{1}{2\delta^2}\pvint_\R \coth\bigg(\frac{\pi(y-x)}{2\delta}\bigg)\bigg\{\pvint_{\R} \csch\bigg(\frac{\pi(z-y)}{\delta}\bigg)f(z)\,\mathrm{d}z+\ii\delta f(z)\bigg\}\mathrm{d}y\\ 
&+\frac{\ii}{\delta}\bigg\{\pvint_{\R} \csch\bigg(\frac{\pi(z-y)}{\delta}\bigg)f(z)\,\mathrm{d}z+\ii\delta f(z)\bigg\}+\ii (\tilde{T}f)(x).
\end{align*}
Using \eqref{hyperbolicidentity} again, we find
\begin{align*}
(\tilde{T}(T-\tilde{T})f)(x)=& -(T(T-\tilde{T})f)(x)-\ii(Tf)(x)+\ii((T-\tilde{T})f)(x)-f(x)+\ii(\tilde{T}f)(x) \\
=&  -(T(T-\tilde{T})f)(x)-f(x)
\end{align*}
and the result \eqref{TplusTminus} follows. 
\end{proof}

\section{Details on the KdV-limit}
\label{app:KdV}
We give details on why the non-chiral ILW equation in \eqref{2ilw} does not have a non-trivial KdV-limit $\delta\to\infty$, as claimed in Section~\ref{subsec1.1}. 

For that, we first recall how one can obtain the KdV equation from the standard ILW equation
\begin{equation}\label{ilw}
u_t+2uu_x+Tu_{xx}=0
\end{equation}
in the limit $\delta\to 0$ \cite{kodama1981,scoufis2005}. 

By inserting the expansion \cite[Eq. A3]{scoufis2005}
 \begin{align}\label{texp}
(Tf)(x)=-\frac1{2\delta}\int\limits_{-\infty}^x f(z)\,\mathrm{d}z+\frac1{2\delta}\int\limits_{x}^{\infty} f(z)\,\mathrm{d}z+\frac{\delta}3 f'(x)+\mathrm{O}(\delta^3) \qquad \text{as } \delta\to 0^+,
\end{align}
we obtain
\begin{equation}\label{ILWdeltalimit1}
u_t+2uu_x-\frac{1}{\delta}u_x+\frac{\delta}{3}u_{xxx}=\mathrm{O}(\delta^5).
\end{equation}
Then, under the transformations \cite[Eq. 1.7]{scoufis2005}
\begin{equation}\label{cov}
u\to \frac{1}{2\delta}+\frac{\delta}{4}u,\qquad x\to 2x,\qquad t\to \frac{24}{\delta}t,
\end{equation}
\eqref{ILWdeltalimit1} becomes
\begin{equation}
\frac{\delta^2}{96}(u_t+6uu_x+u_{xxx})=\mathrm{O}(\delta^3),
\end{equation}
and one obtains the KdV equation in the limit $\delta\to 0$.

For the non-chiral ILW equation, we also need the expansion
\begin{align}\label{ttexp}
(\widetilde{T}f)(x)=-\frac1{2\delta}\int\limits_{-\infty}^x f(z)\,\mathrm{d}z+\frac1{2\delta}\int\limits_{x}^{\infty} f(z)\,\mathrm{d}z-\frac{\delta}6 f'(x)+\mathrm{O}(\delta^3)\qquad  \text{as } \delta\to 0^+,
\end{align}
whose proof is analogous to that of \eqref{texp} in \cite{scoufis2005}. Inserting \eqref{texp}--\eqref{ttexp} into \eqref{2ilw}, we obtain
\begin{align}\label{2ILWdeltalimit}
\begin{split}
&u_t+2uu_x-\frac{1}{\delta}u_x-\frac{1}{\delta}v_x+\frac{\delta}{3}u_{xxx}-\frac{\delta}{6} v_{xxx} =\mathrm{O}(\delta^3), \\
&v_t-2vv_x+\frac{1}{\delta}v_x+\frac{1}{\delta}u_x-\frac{\delta}{3}v_{xxx}+\frac{\delta}{6} u_{xxx} =\mathrm{O}(\delta^3).
\end{split}
\end{align}
In this case, only the first $\delta$-singular term in each equation can be removed by translations $u\to u+1/2\delta$, $v\to v+1/2\delta$. Thus, the task of finding a change of variables analogous to \eqref{cov} such that \eqref{2ILWdeltalimit} has a nontrivial limit as $\delta\to 0$ is complicated by the presence of $\delta$-singular terms. Case-by-case analysis of dominant balance possibilities shows that no such change of variables exists.

\bigskip

\noindent

\noindent {\bf Acknowledgements.} {\it We
thank Junichi Shiraishi for inspiring discussions, and Rob Klabbers for valuable comments and discussions that helped us to improve this paper.
We are grateful to insightful questions and comments of a referee helping us to improve the paper; we also thank this referee for pointing out the references in Section~\ref{subsec1.3}.
BKB acknowledges support from the G\"oran Gustafsson Foundation. EL acknowledges support from the Swedish Research Council, Grant No. 2016-05167, and by the Stiftelse Olle Engkvist Byggm\"astare, Contract 184-0573.
JL is grateful for support from the G\"oran Gustafsson Foundation, the Ruth and Nils-Erik Stenb\"ack Foundation, the Swedish Research Council, Grant No. 2015-05430, and the European Research Council, Grant Agreement No. 682537.}
\bigskip

\nocite{kodama1982}
\nocite{lebedev1987}

\bibliographystyle{unsrt}

\bibliography{BLL3}

\end{document}